\documentclass[british,a4paper,10pt]{article} 

\usepackage{amssymb}
\usepackage{amsmath}
\usepackage{subfigure}
\usepackage{slashed}
\usepackage{amsmath}
\usepackage[normalem]{ulem}
\usepackage{graphicx,psfrag}
\usepackage{mathrsfs}
\usepackage{amsfonts}
\usepackage{multirow}
\usepackage{hyperref}
\usepackage{placeins}
\usepackage{amsthm}
\usepackage{latexsym}
\usepackage[dvips]{epsfig}
\usepackage{enumitem}  
\usepackage{soul}
\usepackage{graphicx} 
\usepackage{amsthm}
\usepackage{latexsym}   
\usepackage[dvips]{epsfig}
\usepackage[utf8]{inputenc}
\usepackage[english]{babel}
\usepackage{multicol}
\usepackage{color}
\usepackage{comment}
\usepackage[dvipsnames]{xcolor}
\usepackage{wasysym}
\usepackage{mathrsfs}
\usepackage{eufrak}
\usepackage{bm}
\usepackage{authblk}
\usepackage{slashed}
\usepackage{yhmath} 
\usepackage{stmaryrd}

\usepackage{lipsum}

\theoremstyle{plain}
\newtheorem{proposition}{Proposition}
\newtheorem{lemma}{Lemma}

\newtheorem{assumption}{Assumption}

\newtheorem{corollary}{Corollary}

\newtheorem{remark}{Remark}

\def\bmeta{{\bm \eta}}
\def\bmg{{\bm g}}
\def\bmsigma{{\bm \sigma}}

\def\ue{\underline{e}}
\def\uL{\underline{L}}
\def\bme{{\bm e}}
\def\bmue{\underline{\bme}}
\def\bmA{{\bm A}}

\newenvironment{IndRemark}{
  \begin{center}
  \begin{minipage}{0.95\textwidth}
    \begin{remark}
      \em
}
{
\end{remark}  
\end{minipage}
\end{center}
}

\begin{document}

\title{\textbf{Spin-0 fields and the NP-constants close to spatial infinity in Minkowski spacetime}}
\author[1]{Edgar Gasper\'in
  \footnote{E-mail address:{\tt edgar.gasperin@tecnico.ulisboa.pt}}}
\author[,2]{Rafael Pinto
   \footnote{E-mail address:{\tt rafael.pastor.pinto@tecnico.ulisboa.pt}}}
  \affil[1,2] {CENTRA, Departamento de F\'isica, Instituto Superior
    T\'ecnico IST, Universidade de Lisboa UL, Avenida Rovisco Pais 1,
    1049 Lisboa, Portugal}

\maketitle
  
\begin{abstract}

  The NP constants of massless spin-0 fields propagating in Minkowski
  spacetime are computed close to spatial and null infinity by means
  of Friedrich's \emph{$i^0$-cylinder}. Assuming certain regularity
  condition on the initial data ensuring that the field extends
  analytically to the critical sets,
  it is shown that the NP constants at future $\mathscr{I}^{+}$ and
  past null infinity $\mathscr{I}^{-}$ are independent of each
  other. In other words, the classical NP constants at
  $\mathscr{I}^{\pm}$ stem from different parts of the initial data
  given on a Cauchy hypersurface.  In contrast, it is shown that,
  using a slight generalisation of the classical NP constants, the
  associated quantities ($i^0$-cylinder NP constants) do not require
  the regularity condition being satisfied and give rise to conserved
  quantities at $\mathscr{I}^{\pm}$ that are determined by the same
  piece of initial data which, in turn, correspond to the terms controlling the
  regularity of the field. Additionally, it is shown how the
  conservation laws associated to the NP constants can be exploited
  to  construct, in flat space, heuristic asymptotic-system expansions
  which are sensitive to the
  logarithmic terms at the critical sets.
\end{abstract}

\textbf{Keywords:} Newman-Penrose constants, asymptotic system,
cylinder at spatial infinity.

\section{Introduction}\label{sec:Intro}

The Newman-Penrose (NP) constants are a set of conserved charges at
null infinity originally introduced in \cite{NewPen68}. These
quantities are 2-surface integrals computed at cuts $\mathcal{C}
\approx \mathbb{S}^2$ of null infinity ($\mathscr{I}$). In the linear
theory (Maxwell's equations and linearised gravity) there is an
infinite hierarchy of these conserved quantities while in the
non-linear theory (General Relativity) only 10 quantities remain
conserved ---see \cite{NewPen68}.  Although the general physical
interpretation of these charges is still a source of debate ---see
\cite{PenRin86, DaiVal02, Bac09}--- these quantities remain conserved
in general asymptotically flat spacetimes even with non-trivial
dynamics such as that of black hole collisions ---see \cite{DaiVal02}.

\medskip

Recently, the analysis of asymptotic quantities, in particular the
Bondi-Metzner-Sachs (BMS) charges, has gained some interest due to its
connection with the concept of black hole soft hair
\cite{HawPerStro16, HawPerStro17, HeLysMitStr15}.  A key element in
these discussions is the relation between asymptotic charges at past
and future null infinity.  The singular behaviour of the conformal
structure at spatial infinity makes the matching problem particularly
difficult.  Hence, the relation between conserved quantities at past
and future null infinity crucially depends on the (partial) resolution
of the singular nature of spatial infinity ($i^0$). In this context,
Friedrich's cylinder representation of spatial infinity is
particularly adapted for this problem. Unlike other conformal
representations adapted to spatial infinity where $i^0$ is mapped to
a point, in Friedrich's conformal representation, $i^0$ is blown-up to
a cylinder $I \approx (-1,1) \times\mathbb{S}^2$. The
regions of spacetime where future and past null infinity meet spatial
infinity correspond to the \emph{critical sets} $I^{\pm} = \{\pm 1\}
\times \mathbb{S}^2$, respectively.  One of the advantages of this
set-up is that
one can relate quantities at the critical sets with
data given on a Cauchy hypersurface $\mathcal{S}$ intersecting the
cylinder at $\{0\} \times \mathbb{S}^2$.  Although Friedrich's
cylinder representation can in principle be constructed for general
asymptotically flat spacetimes through a class of curves known as
\emph{conformal geodesics} (defining the $F$-gauge), in general, it is
hard to write in explicit and closed form the spacetime metric in such
gauge ---except for the case of the Minkowski, de Sitter and anti-de
Sitter spacetimes.  Nonetheless, the study of fields propagating in
the $i^0$-cylinder conformal representation of Minkowski spacetime is
still a valuable model as the degeneracy of the conformal structure at
$i^0$ still has an impact in the regularity of the solution even in
this simplified set-up.  In particular, it has been shown that even in
the case of linear fields (spin-1 and spin-2 fields) with prescribed
analytic data close to $i^0$, the solution develops logarithmic terms
at the critical sets leading to a polyhomogeneous peeling behaviour
---see \cite{Val07, GasVal20, MinMacVal22}.  More recently, the
analysis for the spin-1 and spin-2 fields propagating in this
background was used for the matching problem of the BMS charges at
past and future null infinity in \cite{ValAli22} where the charges are
computed in terms of initial data given on a Cauchy slice.  A result,
similar in spirit but for the NP constants was obtained in earlier in
\cite{GasVal20}.  The aim of this article is to fill the gap for the
case of the scalar wave equation (spin-0) by computing the classical
NP constants and a modified version of them ($i^0$-cylinder NP
constants) which we show correspond to the part of the initial data
that controls the regularity of the field at the critical sets.
Additionally, we show how the main identities behind the definition of
the NP constants can be used to obtain a generalisation of the
asymptotic system heuristics of \cite{DuaFenGasHil21} which is
sensitive to the $i^0$-cylinder logs. This represents the
``higher-order'' generalisation of \cite{DuaFenGasHil22b} in flat
spacetime. Further generalisations for asymptotically flat spacetimes
will be pursued in future work.

\section{Conformal methods and the $i^0$-cylinder representation}\label{sec:Conf}

In Penrose's conformal approach \cite{Pen63}, the fall-off of the
gravitational field and the asymptotic structure of spacetime is to be
studied not in terms of the \emph{physical} spacetime
$(\tilde{\mathcal{M}},\tilde{\bmg})$ ---satisfying the Einstein field
equations $\tilde{G}_{ab}=\tilde{T}_{ab}$--- but rather in terms of an
conformally related spacetime $(\mathcal{M},\bmg)$ ---which will be
called the \emph{unphysical} spacetime--- where
$\bmg=\Omega^2\tilde{\bmg}$.  The conformal factor $\Omega$ serves as
a boundary defining function in the sense that $\mathscr{I}$ is
defined as the set of points in the unphysical manifold for which
$\Omega=0$ but $\mathbf{d}\Omega \neq 0$. A distinguished point in the
conformal structure of asymptotically flat spacetimes is spatial
infinity ($i^0$) for which $\Omega=0$ and $\mathbf{d} \Omega=0$ ---see
\cite{Val16, Fra04, Fri02CEE} for an exhaustive discussion on the
conformal approach and the conformal Einstein field equations.
Following these naming conventions for the physical/unphysical
spacetime, in subsection \ref{sec:i0cylinder} the $i^0$-cylinder
conformal representation of the Minkowski spacetime is revisited and
in subsection \ref{sec:Physical_NP_F_frames} the relation between the
Newman-Penrose, Friedrich and physical null frames is given.

\subsection{The $i^0$-cylinder in  Minkowski spacetime}\label{sec:i0cylinder}

Let $(\tilde{t},\tilde{\rho},\vartheta^A)$ with $A=1,2$ denote
\emph{physical} spherical polar coordinates where $\vartheta^A$
represent some coordinates on $\mathbb{S}^2$. In these coordinates the
\emph{physical} Minkowski metric $\tilde{\bmeta}$ reads
\begin{align}\label{eq:MinkowskiMetricPhysicalPolar}
\tilde{\bmeta}=-\mathbf{d}\tilde{t}\otimes\mathbf{d}\tilde{t}
+\mathbf{d}\tilde{\rho}\otimes \mathbf{d}\tilde{\rho}+\tilde{\rho}^2
\mathbf{\bm\sigma},
\end{align}
with~$\tilde{t}\in(-\infty, \infty)$, $\tilde{\rho}\in [0,\infty)$
  where~$\bm\sigma$ denotes the standard metric on~$\mathbb{S}^2$.  As
  an intermediate step to obtain the conformal representation we are
  interested in, let us introduce \emph{unphysical} spherical polar
  coordinates $(t, \rho, \vartheta^A)$ given by
\begin{align}
 \label{eq:physicalToUnphysicaltrho}
   t=\frac{\tilde{t}}{\tilde{\rho}^2-\tilde{t}^2}, \qquad \rho =
   \frac{\tilde{\rho}}{\tilde{\rho}^2-\tilde{t}^2}.
\end{align}
Then, expressing the physical Minkowski metric $\tilde{\bmeta}$ in the
unphysical spherical polar coordinates $(t, \rho, \vartheta^A)$ one
can readily identify the \emph{inversion} conformal representation of
the Minkowski spacetime $(\mathbb{R}^4,\bmeta)$ where
\begin{align}
 \label{eq:InverseMinkowskiMetricDef}
 \bmeta=\Xi^2 \hspace{0.5mm}\tilde{\bmeta},
\end{align}
with
\begin{align}
\label{eq:InverseMinkowskiUnphysicaltrhocoords}
\bmeta=-\mathbf{d}t\otimes\mathbf{d}t +\mathbf{d}\rho\otimes
\mathbf{d}\rho+\rho^2 \mathbf{\bm\sigma}, \qquad \Xi=\rho^2 -t^2,
\end{align}
where~$t\in(-\infty,\infty)$ and~$\rho\in~[0,\infty)$.  In this
  conformal representation, spatial infinity and the origin are
  interchanged as $i^0$ is represented by the point
  $(t=0,\rho=0)$ in $(\mathbb{R}^4,\bmeta)$---see \cite{Ste91, Val16,
    GasVal20, FenGas23}.  Introducing coordinates
  $(\tau,\rho,\vartheta^A)$, where $t= \rho \tau$ and considering the
  conformal metric $\bmg= \rho^{-2}\bmeta$ one obtains
\begin{align*}
\bmg=- \mathbf{d}\tau\otimes \mathbf{d}\tau
+\frac{(1-\tau^2)}{\rho^2}\mathbf{d}\rho \otimes \mathbf{d}\rho -
\frac{\tau}{\rho} \mathbf{d}\rho\otimes \mathbf{d}\tau -
\frac{\tau}{\rho} \mathbf{d}\tau \otimes \mathbf{d}\rho + \bmsigma,
\end{align*}
where the \emph{unphysical metric} $\bmg$ is related to the physical
metric via
\begin{align}
  \bmg=\Theta^2\tilde{\bm\eta}, \qquad \text{where} \qquad
  \Theta := \frac{\Xi}{\rho}=\rho (1-\tau^2).
\end{align}
The unphysical metric $\bmg$ will be called the $i^0$-cylinder
metric for short and the coordinates $(\tau, \rho, \vartheta^A)$
referred to as the $F$-coordinate system ---see \cite{MinMacVal22, GasVal20, ValAli22}.
The relation between the physical coordinates and the $F$-coordinates is given by
\begin{align}\label{eq:Ftophys}
  \tau = \frac{\tilde{t}}{\tilde{\rho}}, \qquad \rho =
  \frac{\tilde{\rho}}{\tilde{\rho}^2-\tilde{t}^2}.
\end{align}
The inverse transformation can be succinctly written as
\begin{align}
  \tilde{t}=\frac{\tau}{\Theta}, \qquad \tilde{\rho}=
  \frac{1}{\Theta}.
\end{align}
Additionally, notice that the physical retarded and advanced times,
defined as $\tilde{u}:=\tilde{t}- \tilde{\rho}$ and $ \tilde{v}:= \tilde{t}+
\tilde{\rho}$ respectively,
are related to the unphysical advanced and retarded times via
\begin{align}\label{eq:UnphysPhysAdvRet}
v:=t-\rho=-\rho(1-\tau)=\tilde{v}^{-1}, \qquad u:=t+\rho=\rho(1+\tau)=-\tilde{u}^{-1}.
\end{align}
In this conformal representation of the Minkowski spacetime
future and past null infinity are located at
\begin{align}
 \mathscr{I}^{+} \equiv \{ p \in \mathcal{M} \; \rvert\; \tau(p) =1
 \}, \qquad \mathscr{I}^{-} \equiv \{ p \in \mathcal{M} \; \rvert \;
 \tau(p) =-1\}.
\end{align}
The name $i^0$-cylinder comes from the fact that spatial infinity
gets mapped to an extended set $I \approx \mathbb{R}\times \mathbb{S}^2$
\begin{align*}
 I \equiv \{ p \in \mathcal{M} \; \rvert \;\; |\tau(p)|<1, \;
 \rho(p)=0\}, \qquad I^{0} \equiv \{ p \in \mathcal{M}\; \rvert \;
 \tau(p)=0, \; \rho(p)=0\},
\end{align*}
and the region where spatial and null infinity meets is
represented by
\begin{align*}
 I^{+} \equiv \{ p\in \mathcal{M} \; \rvert \; \tau(p)=1, \; \rho(p)=0
 \}, \qquad I^{-} \equiv \{p \in \mathcal{M}\; \rvert \; \tau(p)=-1,
 \; \rho(p)=0\},
\end{align*}
which are called the critical sets.

\subsection{The physical, the NP and the F-frames}\label{sec:Physical_NP_F_frames}

As it will be used in the up-coming discussion, 
we introduce  the following adapted $\bmg$-null frame:
\begin{align}\label{eq:Fframe}
 \bme
  =(1+\tau)\bm\partial_{\tau}{} -
  \rho\bm\partial_{\rho}{},  \qquad
  \bmue =(1-\tau)\bm\partial_{\tau}{} +
  \rho\bm\partial_{\rho}{}, \qquad
  \bme_{\bmA } \qquad \text{with} \qquad  \bmA = \{\uparrow, \downarrow\},
\end{align}
where the $\bme_{\bmA}$ is a complex null frame on $\mathbb{S}^2$ with associated
coframe $\bm\omega^{\bmA}$ such that the standard metric on  $\mathbb{S}^2$
reads
\begin{align}
\bm\sigma=2(\bm\omega^{\uparrow}\otimes
\bm\omega^{\downarrow}+\bm\omega^{\downarrow}\otimes \bm\omega^{\uparrow}).
\end{align}
\begin{IndRemark}
  As NP-frames (hinged at $\mathscr{I}^{\pm}$, respectively) will be
  employed, the symbols $\pm$ will be used to identify such
  frames. Hence, to avoid confusion, we have labelled the elements of
  the frame on $\mathbb{S}^2$ with the symbols $\uparrow \downarrow$.
\end{IndRemark}
The tetrad  $\{\bme, \bmue, \bme_{\bmA}\}$  will be called the $F$-frame to
distinguish it from other frames that will be introduced later on.
Observe that, in terms of the $F$-frame,
the unphysical metric $\bmg$ can be expressed as
\begin{align}\label{eq:UnphysicalMetricNullTetrad}
g_{ab}=e_{(a}\ue_{b)} -  \omega^{\uparrow}_{(a}\omega^{\downarrow}_{b)}\,,
\end{align}
so that the normalisation of the tetrad
is~$e_a\ue^a=-\omega^{\uparrow}_ae_{\downarrow}^a=-2$ while all the other
contractions vanish.
Similarly, the $\tilde{\bmeta}$-null frame denoted as $\{L, \; \uL,\;
\tilde{\bme}_{\bmA}\}$ and given by
\begin{align}
L = \bm\partial_{\tilde{t}} +\bm\partial_{\tilde{\rho}}, \qquad \uL =
\bm\partial_{\tilde{t}}-\bm\partial_{\tilde{\rho}}, \qquad
\tilde{\bme}_{\bmA}=\tilde{\rho}^{-1}\bme_{\bmA},
\end{align}
will be referred to as the physical null frame.  With exception of the
physical null vectors $L,\uL$ all the quantities referring to the
physical spacetime will be decorated with a tilde over the main symbol. Hence,
quantities such as $\tilde{\phi}$ and $\phi$ correspond to
physical and unphysical (conformally rescaled) fields respectively.

\medskip
Another unphysical frame that will play an important role in the
definition of the NP constants is the NP-frame. The relation between
the F-frame and the NP-frame in Minkowski spacetime was derived and
discussed in \cite{GasVal20} ---see also \cite{ValAli22}.  Using the
results in \cite{GasVal20} and the expressions described in this
section, one obtains the following proposition giving the relation
between these three frames:

 \begin{proposition}\label{Prop:NPtoFgauge}
   The NP-frame (hinged at $\mathscr{I}^{\pm}$), the $F$-frame and the
   standard physical frame for the Minkowski spacetime are related via:
\begin{align*}
  \text{\emph{NP hinged at} $\mathscr{I}^{+}$}:& \quad\bme^{+} =
  4(\Lambda_{+})^{2} \bme = \Theta^{-2} L, && \bmue^{+}=
  \tfrac{1}{4}(\Lambda_{+})^{-2}\bmue = \uL, && \bme_{\bmA}^{+}=
  \bme_{\bmA}= \Theta^{-1}\tilde{\bme}_{\bmA}\\ \text{\emph{NP hinged
    at} $\mathscr{I}^{-}$}:& \quad\bme^{-} =
  \tfrac{1}{4}(\Lambda_{-})^{-2} \bme = L, && \bmue^{-}=
  4(\Lambda_{-})^{2}\bmue = \Theta^{-2} \uL, && \bme_{\bmA}^{-}=
  \bme_{\bmA}= \Theta^{-1}\tilde{\bme}_{\bmA}.
\end{align*}
where the conformal factor $\Theta$ and boost parameter $\kappa$
in~$F$-coordinates and physical coordinates read, respectively,
\begin{align}\label{eq:CF-thetaAndBoostParameter}
  \Theta := \rho (1-\tau^2) = \frac{1}{\tilde{\rho}}, \qquad \varkappa
  := \frac{1+\tau}{1-\tau} = -\frac{\tilde{v}}{\tilde{u}}.
\end{align}
Similarly, the Lorentz transformation relating the NP and F-frames is
encoded in the quantities
\begin{align}\label{eq:LorentzTransf}
  (\Lambda_{+})^{2}:= \Theta^{-1}\varkappa^{-1}=
  \rho^{-1}(1+\tau)^{-2}, && (\Lambda_{-})^{2}:= \Theta^{-1}\varkappa=
  \rho^{-1}(1-\tau)^{-2}.
\end{align}
 \end{proposition}
 \begin{IndRemark}
    Notice that the NP-frame is not a null tetrad for the physical
    metric $\tilde{\bmeta}$, but rather respect to a conformally
    related metric $\bmg'=\vartheta^2\tilde{\bmeta}$ for some
    conformal factor $\vartheta$. In the particular case of the
    Minkowski spacetime it turns out that such conformally related
    metric $\bmg'$ and the $i^0$-cylinder metric $\bmg$
    coincide.  Namely, in general $\bmg' =\kappa^2 \bmg$, with
    $\kappa=1$ in the particular case of the Minkowski
    spacetime. Formal asymptotic expansions for the conformal
    transformation $\kappa$ and the Lorentz transformation relating
    the NP and F-gauges have been computed for time-symmetric initial
    data of asymptotically flat spacetimes in \cite{FriKan00}.
 \end{IndRemark}

\section{Spin-2 and Spin-0 fields as models for (linearised) gravity}\label{sec:ModelEquations}

Curvature oriented formulations such as that of Newman \& Penrose's
\cite{NewPen62}, Christodoulou \& Klainerman's \cite{ChrKla93} or
H. Friedrich's conformal Einstein field equations \cite{Fri81,
  Fri02CEE}, have at its core equations for the Weyl tensor derived
from the Bianchi identities.  These type of formulations are in stark
contrast with purely metric formulations such as the ADM, Z4
and generalised harmonic gauge formulations.  Hence, it is natural
that when looking for suitable linear models for the field equations
of General Relativity that one encounters different equations
better adapted to each type of formulation.  For instance, in
Penrose's take \cite{PenRin86} on linearised gravity, (in vacuum and
around flat spacetime) the gravitational field is modelled by a spin-2
field satisfying
\begin{align}\label{eq:spin-2}
  \tilde{\nabla}^{AA'}\tilde{\phi}_{ABCD}=0,
\end{align}
where $\tilde{\phi}_{ABCD}$ is a totally symmetric spinor field which
represents the linear counterpart of the Weyl spinor ---see
\cite{Val03a}.  This equation has the advantage of being conformally
invariant and also serves as an ideal linear model for the Bianchi
sector of the conformal Einstein field equations.  A similar remark
holds true for the Maxwell's equations in flat spacetime which can be
described in terms of the spin-1 equation
\begin{align}\label{eq:spin-2}
  \tilde{\nabla}^{AA'}\tilde{\phi}_{AB}=0,
\end{align}
where $\tilde{\phi}_{AB}$ is a symmetric spinor (the Maxwell spinor)
encoding the electromagnetic field ---see \cite{ValAli22}.  In
\cite{NewPen68} the NP constants were defined in the linear set-up for
the spin-1 and spin-2 fields and in the non-linear case for the
Einstein field equations written in the Newman-Penrose formalism.  In
\cite{GasVal20} and \cite{ValAli22} the framework of the
$i^0$-cylinder was used to compute the NP constants and the BMS
charges, respectively, for these fields propagating in flat spacetime
close to null and spatial infinity.  Although, the spin-2 field is a
good model for curvature oriented formulations, for metric
formulations, wave equations are better suited.  For instance, the
standard formulation of linearised gravity ---see \cite{Mag07a,
  Wal84a}--- is usually presented in an way that mimics that of the
hyperbolic reduction of the Einstein field equations in harmonic
gauge, in which the final expression is a wave equation for the metric
components.  Following the previous naming conventions, one could call
scalar fields satisfying these wave equations as spin-0 fields. In
this section, the \emph{physical} wave equation in flat spacetime will
be written as a wave equation for an \emph{unphysical} field
propagating in the $i^0$-cylinder background. A way to solve the
resulting equation, in a similar approach as that of the spin-1 and
spin-2 fields in \cite{Val07, GasVal20}, has been described in
\cite{MinMacVal22}.  In the remainder of this section, the results of
\cite{MinMacVal22} for the spin-0 field will be revisited so that they
can be used in the next section to compute the NP constants.

\medskip

Recall that given two conformally related spacetimes
$(\tilde{\mathcal{M}},\tilde{\bmg})$ and $(\mathcal{M}, \bmg)$ where
$\bmg=\Omega^2\tilde{\bmg}$, the conformal transformation formula for
their D'Alembertian operators
is given by
\begin{align}
  \square \phi- \frac{1}{6} \phi R = \Omega ^{-3} \bigg(
  \tilde{\square }\tilde{\phi}- \frac{1}{6} \tilde{\phi}
  \tilde{R}\bigg), \label{eq:General_Wave_Conformal_Transformation}
\end{align}
where $\tilde{\square}=\tilde{g}^{ab}\tilde{\nabla}_a\tilde{\nabla}_b$
and $\square=g^{ab}\nabla_a\nabla_b$ with $\nabla$ and 
$\tilde{\nabla}$, ~$R$ and~$\tilde{R}$,
denoting the Levi-Civita connections and Ricci
scalars of~$\bmg$ and~$\tilde{\bmg}$,
respectively.  Let $\tilde{\phi}$ be a scalar propagating in
flat spacetime $(\mathbb{R}^4, \tilde{\bmeta})$ according to:
\begin{align} \label{eq:Physical_wave}
   \tilde{\square} \tilde{\phi} = 0.
\end{align}
If one chooses as conformal transformation that of the $i^0$-cylinder,
using equation \eqref{eq:General_Wave_Conformal_Transformation}, the
unphysical equation simply reads:
\begin{align}\label{eq:Unphysical_wave}
  \square \phi =0.
\end{align}
Observe that this is a consequence of the fact that for flat spacetime
and for its $i^0$-cylinder conformal representation, one has that
$R=\tilde{R}=0$ ---see Remark \ref{remark:singularvsregular}.

\begin{IndRemark}\label{remark:singularvsregular}
  Although for vacuum spacetimes (with vanishing cosmological
  constant) equation \eqref{eq:General_Wave_Conformal_Transformation}
  is \emph{formally regular}, in general, this is not the case. In
  other words, for generic spacetime backgrounds \emph{formally
  singular} coefficients such as $\Omega^{-3}$ would persist.  Despite
  that for equations describing the propagation of fields on fixed
  backgrounds one could potentially exploit the decay of $\tilde{R}$
  in an explicit way, in the non-linear case of the Einstein field
  equations most formulations that include $\mathscr{I}$ are formally
  singular (e.g., \cite{VanHus16}, \cite{DuaFenGasHil22a} and
  \cite{MonRin08}) ---except for H. Friedrich's formulation in
  \cite{Fri81}.
\end{IndRemark}

\begin{IndRemark}
  Since the conformal factor for the $i^0$-cylinder representation of
  the Minkowski spacetime expressed in physical coordinates is just $
  \Theta=\tilde{\rho}^{-1}$, in fact, the unphysical (conformal) field
  $\phi=\Theta^{-1}\tilde{\phi}$ simply corresponds to the
  \emph{radiation field} $\tilde{\rho}\tilde{\phi}$.  However, for a
  different conformal transformation this is not the case. For
  instance, in the standard textbook 
  compactification of Minkowski spacetime
  used to construct its Penrose diagram ---see \cite{GriPod09, Ste91}, the
  conformal factor expressed in physical coordinates is:
  $\Omega=(1+\tilde{u}^2)^{-1/2}(1+\tilde{v}^2)^{-1/2}$. Hence, the
  associated unphysical field $\Omega^{-1}\tilde{\phi}$ would not
  correspond to the radiation field.
\end{IndRemark}

\noindent Equation \eqref{eq:Physical_wave} expressed in $F$-coordinates
explicitly reads

\begin{align}\label{eq:UnphysicalWaveExplicit}
  (\tau ^2-1) \partial _{\tau}^2 \phi -2 \rho \tau \partial
  _{\tau}\partial _{\rho}\phi +  \rho ^2 \partial _{\rho}^2\phi
  + 2 \tau \partial _{\tau}\phi +
  \Delta _{\mathbb{S}^{2}{}}{}\phi   = 0,
\end{align}
where~$\Delta _{\mathbb{S}^{2}{}}{}$ is the Laplace operator on
$\mathbb{S}^2$.
Following the approach taken for the
analysis of the spin-1, spin-2 and spin-0
fields in \cite{ValAli22, MinMacVal22},
one takes the following
Ansatz:
\begin{align}\label{eq:Ansatz}
\phi=
\sum_{p=0}^{\infty}\sum_{\ell=0}^{p}\sum_{m=-\ell}^{m=\ell}
\frac{1}{p!}a_{p;\ell,m}(\tau)Y_{\ell
  m}\rho^{p}.
\end{align}
The use of this Ansatz, in particular,
restricts the initial data to be analytic
close to $i^0$ since
\begin{align}\label{eq:ID_field}
\phi|_{\mathcal{S}} =
\sum_{p=0}^{\infty}
\sum_{\ell=0}^{p}\sum_{m=-\ell}^{m=\ell}\frac{1}{p!}a_{p;\ell,m}(0)Y_{\ell
  m}\rho^p, \qquad \dot{\phi}|_{\mathcal{S}} =
\sum_{p=0}^{\infty}\sum_{\ell=0}^{p}\sum_{m=-\ell}^{m=\ell}
\frac{1}{p!}\dot{a}_{p;\ell,m}(0)Y_{\ell
  m}\rho^p,
\end{align}
where the over-dot denotes a derivative respect to $\partial_\tau$.
As we will see later, despite that the initial data is analytic, in
general, the solution will not.  A calculation ---see
\cite{MinMacVal22}--- shows that solving the wave equation
\eqref{eq:UnphysicalWaveExplicit} reduces to solve the following
ordinary differential equation (ODE) for every $p;\ell,m$

\begin{align}\label{eq:ODE_wave_JacobiPoly}
(1-\tau^2)\ddot{a}_{p;\ell,m} +
  2\tau(p-1)\dot{a}_{p,\ell,m}+(\ell+p)(\ell-p+1){a}_{p;\ell,m}=0.
\end{align}

\noindent The solution of this ODE is given in the following:

\begin{lemma}[wave equation on the $i^0$-cylinder
    background~\cite{MinMacVal22}]\label{Lemma:Sol_Jacobi_and_Logs}
  The solution to equation \eqref{eq:ODE_wave_JacobiPoly} is given
   by:
  \begin{enumerate}[label=(\roman*)]
  \item For $p\geq 1$   and $0\leq \ell \leq p-1$
    \begin{align}\label{eq:Sol_jac_poly}
      a(\tau)_{p;\ell,m} =A_{p,\ell,m}
      \bigg(\frac{1-\tau}{2}\bigg)^{p}
      J_{\ell}^{(p,-p)}(\tau) + B_{p,\ell,m}
      \bigg(\frac{1+\tau}{2}\bigg)^{p}J_{\ell}^{(-p,p)}(\tau)
    \end{align}
  \item For  $p\geq 0$   and $\ell=p$
    \begin{align}\label{eq:Sol_highestharmonic}
      {a}_{p;p,m}(\tau) = \bigg(\frac{1-\tau}{2}\bigg)^{p}
      \bigg(\frac{1+\tau}{2}\bigg)^{p}\Bigg(C_{p,p,m} +D_{p,p,m}
      \int_{0}^{\tau} \frac{ds}{(1-s^2)^{p+1}}\Bigg)
    \end{align}
    where~$J^{\alpha,\beta}_{\gamma}(\tau)$ are the Jacobi polynomials
    and~$A_{p,\ell,m}$, $B_{p,\ell,m}$, $C_{p,p,m} $ and $D_{p,p,m} $
    are constants determined (algebraically) in terms of the
    initial data~$a_{p;\ell,m}(0)$ and $\dot{a}_{p;\ell,m}(0)$.
  \end{enumerate}
\end{lemma}

The interesting feature common to the study of evolution equations
using the $i^0$-cylinder framework is that logarithmic terms appear at
null infinity in the solution. This can be see from inspecting the
hypergeometric function appearing in equation
\eqref{eq:Sol_highestharmonic} for a few values of $p$. For instance
for $p=0$ and $p=1$ one has:
\begin{align}
  {a}_{0;0,0}(\tau) & = C_{000} + \tfrac{1}{2} D_{000} (\log(1 + \tau
  )- \log(1 - \tau )),
  \\ {a}_{1;1,m}(\tau) & = \tfrac{1}{4} (1 - \tau )
  (1 + \tau ) (C_{11m} + \tfrac{1}{4} D_{11m} ( \log(1 + \tau ) -
  \log(1 - \tau ) + 2\tau(1-\tau^2))).
\end{align}
These logarithmic terms impact the linear version of the associated
peeling property ---see \cite{Val07,MinMacVal22}.  In the full
non-linear case further obstructions to the smoothness of null
infinity arise ---see \cite{Val04}.  In the gravitational case one can
place conditions on the initial data condition written ensuring the
absence of these logarithmic terms in the development ---see
\cite{Fri98a}.  For the spin-0 case the analogous condition is given
by:

\begin{IndRemark}\label{Remark:logfreeRemark}(Regularity condition~\cite{MinMacVal22}).
  \emph{Lemma \ref{Lemma:Sol_Jacobi_and_Logs} implies that if initial
  data is chosen so that $\dot{a}_{p;p,m}(0) =0$ then $D_{p;p,m}=0$, and consequently, the
  solution extends analytically to the critical sets.}
\end{IndRemark}

\section{The NP-constants for the spin-0 field close to the $i^0$-cylinder}
\label{sec:NPConstants}

The original constants of Newman and Penrose \cite{NewPen68} are a set
of quantities computed by performing 2-surface integrals at null
infinity of the certain derivatives of the field.  These quantities
are conserved ---hence the name constants--- in the sense that their
value is independent of the cut of null infinity on which they are
computed. In \cite{NewPen68} it was shown for the case of spin-1 and
spin-2 fields (Faraday and linearised Weyl tensors) in Minkowski
spacetime that there is an infinite hierarchy of conserved quantities
while for the full non-linear theory (Weyl tensor) only 10 quantities
remain conserved. In both linear (spin-1 and spin-2) and non-linear
(gravitational) cases, the NP constants arise from a set of asymptotic
conservation laws. For the spin-0 field in flat spacetime ---see
\cite{Keh21_a}--- one has
\begin{align}\label{eq:cons_laws}
 \uL (\tilde{\rho}^{-2\ell}L (\bme^{+})^{\ell}\phi_{\ell m}) = 0,
 \qquad L(\tilde{\rho}^{-2\ell} \uL (\bmue^{-})^{\ell}\phi_{\ell m}) =
 0
\end{align}
where $\phi_{\ell m}= \int_{\mathbb{S}^2} \phi \; Y_{\ell m} \;
d\sigma$ with $d\sigma$ denoting the area element in $\mathbb{S}^2$.
A derivation of a slight generalisation of these these identities is
given in Appendix \ref{App:A}.  Equation \eqref{eq:cons_laws} gives an
infinite hierarchy of exact conservation laws.  Moreover, one can
define the $f(\tilde{\rho})$-modified NP-constants ---see
\cite{GajKeh22}, as follows:
\begin{align}\label{eq:DefModifiedNP}
    {}^{f}\mathcal{N}^{+}_{\ell,m}:= f(\tilde{\rho})L (\bme^{+})^{\ell}\phi_{\ell
      m} \Big|_{\mathcal{C}^{+}}, \\ {}^{f}\mathcal{N}^{-}_{\ell,m}:=
    f(\tilde{\rho})\uL (\bmue^{-})^{\ell}\phi_{\ell
      m}\Big|_{\mathcal{C}^{-}},
  \end{align}
 where $\mathcal{C^{\pm}} \approx \mathbb{S}^2$ is a cut of
 $\mathscr{I}^{\pm}$.  For the case $f(\tilde{\rho})=\tilde{\rho}^2$
 these are called the \emph{classical NP-constants}, denoted as
 $\mathcal{N}^{\pm}_{\ell,m}$, and can be written succinctly as
\begin{align}
\mathcal{N}^{+}_{\ell,m}:= (\bme^{+})^{\ell+1}\phi_{\ell m}
\Big|_{\mathcal{C}^{+}},\label{eq:classicalNP}\\ \mathcal{N}^{-}_{\ell,m}:=
(\bmue^{-})^{\ell+1}\phi_{\ell m} \Big|_{\mathcal{C}^{-}}.
\label{eq:classicalNPMinus}
\end{align}
In the notation for ${}^{f}\mathcal{N}^{+}_{\ell,m}$, in the following,
it will be implicitly understood that $m=-\ell, \dots,0,\dots, \ell$.

\subsection{The classical NP-constants}\label{sec:ClassicalNP_main}

In this section the classical NP constants are computed in terms of
the initial data expressed through the constant parameters appearing
in Lemma \ref{Lemma:Sol_Jacobi_and_Logs}.  It is instructive to
compute the first few NP-constants before discussing the general case.
 \subsubsection{The first few NP constants}
In this subsection the $\ell=0$, $\ell=1$
classical NP-constants at $\mathscr{I}^{+}$ are computed.
From expression \eqref{eq:Ansatz} it follows that
\begin{align}\label{eq:exp_phi_lm}
\phi_{\ell m}=
\sum_{p=\ell}^{\infty}\frac{1}{p!}a_{p;\ell,m}(\tau)\rho^{p}.
\end{align}
Hence, for $\ell=0$ it is enough to compute $\bme^{+}(\phi_{00})$.
Using Proposition \ref{Prop:NPtoFgauge} and equation \eqref{eq:Fframe}
one has
\begin{align}\label{eq:bmeplus1philmraw}
  \bme^{+}(\phi_{\ell m})=
  4(\Lambda_{+})^{2}\sum_{p=0}^{\infty}\frac{1}{p!}\bme(a_{p;\ell,m}(\tau)\rho^{p})
  = 4 \rho^{-1}(1+\tau)^{-2}\sum_{p=0}^{\infty} \frac{1}{p!}\rho^p
  ((1+\tau)\dot{a}_{p;\ell,m}-p a_{p;\ell,m}).
\end{align}
Let
\begin{align}\label{eq:defQ0}
  Q^{0}_{p;\ell,m}(\tau):=(1+\tau)\dot{a}_{p;\ell,m}-p a_{p;\ell,m}.
\end{align}
Using this definition one can express $\bme^{+}(\phi_{\ell m})$ as
\begin{align}\label{eq:bmeplus1philm}
  \bme^{+}(\phi_{\ell m}) = 4 (\Lambda_{+})^{2}\sum_{p=0}^{\infty}
  \frac{1}{p!}\rho^{p}Q^{0}_{p,\ell,m}(\tau).
\end{align}
To obtain the $\ell=0$ NP-constant at $\mathscr{I}^{+}$ one needs to
evaluate $\bme^{+}(\phi_{00})$ at a cut $\mathcal{C}^{+} \subset
\mathscr{I}^{+}$.  Using equation \eqref{eq:bmeplus1philm}, Lemma
\ref{Lemma:Sol_Jacobi_and_Logs} gives
\begin{align}\label{N00constExpression}
  \mathcal{N}^{+}_{0,0}= \lim_{\substack{\rho \to \rho_{\star} \\ \tau
      \to 1}} \bme^{+}(\phi_{00}) = \sum_{p=0}^{\infty}
  \frac{1}{p!}\rho^{p-1}_{\star}Q^{0}_{p,0,0}|_{\mathscr{I}^{+}},
\end{align}
where $\rho_{\star}$ is a constant that parametrises the cut
$\mathcal{C}^{+}$ and $Q^{0}_{p,\ell,m}|_{\mathscr{I}^{+}}:=\lim_{\tau
  \to 1}Q^{0}_{p,\ell,m}(\tau)$.  In particular, $\rho_{\star}=0$
corresponds to the choice $\mathcal{C}^{+}= I^{+}$.  A direct
calculation using Lemma \ref{Lemma:Sol_Jacobi_and_Logs} shows that in
fact $Q^{0}_{0,0,0}(\tau)=D_{000}/(1-\tau)$. Therefore, if the
regularity condition of Remark \ref{Remark:logfreeRemark} is not
satisfied, then the classical NP constants are not well-defined
regardless of the cut at which they are evaluated.  Hence, to compute
the classical $\ell=0$ NP-constant the regularity condition needs to
be imposed.  Moreover, once the regularity condition is imposed, then
the value of the classical $\ell=0$ NP constant is independent of the
cut.  To see this, we have the following
\begin{IndRemark}\label{rem:l0}
   A direct calculation
using Lemma \ref{Lemma:Sol_Jacobi_and_Logs}
gives
  \begin{subequations}\label{eq:rem:l0}
  \begin{align}
    Q^{0}_{0;0,0}(\tau)&=D_{000}(1-\tau)^{-1},\label{rem:l0:eq1} \\
    Q^{0}_{p;0,0} (\tau)&=-2^{1-p}pA_{p,0,m}(1-\tau)^{p-1} \qquad \text{ for}\qquad p\neq 0.
    \label{rem:l0:eq2}
  \end{align}
 \end{subequations}
\end{IndRemark}
Hence,
using Remark \ref{rem:l0} and assuming the regularity condition is satisfied, then
one has
\begin{align}
\mathcal{N}^{+}_{0,0}=-A_{1,0,0}
\end{align}

The fact that the regularity condition needs to be imposed to have
well-defined (classical) NP constants is expected since a similar
observation holds in the spin-1 and spin-2 cases where the analogous
regularity condition ($D_{\ell \ell m}=0$ in the present case) was
imposed via ``Assumption 2'' in \cite{GasVal20} ---see also
\cite{Val98, Val99a}.  Moreover, observe that Remark \ref{rem:l0}
shows that the only contributing term is that with $p=1$.  As it will
be shown later, for the next constant $\mathcal{N}^{+}_{1,m}$ only the
$p=2$ term contributes. A similar pattern was observed in
\cite{GasVal20} for the spin-1 and spin-2 cases.  Therefore, this
behaviour was expected. Additionally, if other than the $p=1$ term
contributes to equation \eqref{N00constExpression} that would mean
that $\mathcal{N}^{+}_{0,0}$ depends on the value of $\rho_{\star}$
and hence the cut where it is evaluated and thus contradicting the
constancy of this quantity ---see \cite{NewPen68}.  In other words, if
one assumes that the NP constants are \emph{well-defined} then the
detailed form of the expressions in Remark \ref{rem:l0} is not
required to assert which part of the initial data parameters of Lemma
\ref{Lemma:Sol_Jacobi_and_Logs} determines the value of
$\mathcal{N}^{+}_{0,0}$ as long as only the parameter $A_{p,\ell,m}$
---and not $B_{p,\ell,m}$--- appears in equation \eqref{rem:l0:eq2}.

\medskip

To determine the $\ell=1$- NP constants we need to compute
$(\bme^{+})^{2}(\phi_{1m})$.  Using Proposition \ref{Prop:NPtoFgauge},
and assisted with the previous calculation, one has
\begin{align}\label{eq:bmeplus2philmraw}
  (\bme^{+})^2(\phi_{\ell m})= 4^2(\Lambda_{+})^{2}\bme \Big(
  (\Lambda_{+}{})^2 \sum_{p=\ell}^{\infty}
  \frac{1}{p!}\rho^{p}Q^{0}_{p,\ell,m}(\tau)\Big).
\end{align}
An explicit calculation, using equations \eqref{eq:Fframe} and
\eqref{eq:LorentzTransf} gives
$\bme((\Lambda_{+})^{2})=-(\Lambda_{+})^{2}$.  Using the latter
renders
\begin{align}\label{eq:defQ1}
  (\bme^{+})^2(\phi_{\ell m})= 4^2(\Lambda_{+})^{4}
  \sum_{p=\ell}^{\infty} \frac{1}{p!}\rho^{p}
  \big((1+\tau)\dot{Q}^{0}_{p,\ell,m}(\tau)-(p+1){Q}^{0}_{p,\ell,m}(\tau)\big).
\end{align}
Defining $Q^{1}_{p,\ell,m}(\tau):=
\big((1+\tau)\dot{Q}^{0}_{p,\ell,m}(\tau)-(p+1){Q}^{0}_{p,\ell,m}(\tau)\big)$
and using equation \eqref{eq:defQ0} gives
\begin{align}\label{eq:DefQ1IntermsQ0}
  Q^{1}_{p,\ell,m}=(1+\tau)^2\ddot{a}_{p;\ell,m}-2p(1+\tau)\dot{a}_{p;\ell,m}
  + p(p+1)a_{p;\ell,m}.
\end{align}
Hence one can compute $\mathcal{N}^{+}_{1,m}$ as
\begin{align}\label{NP1_precomputation}
  \mathcal{N}^{+}_{1,m}= \lim_{\substack{\rho \to \rho_{\star} \\ \tau
      \to 1}} (\bme^{+})^2(\phi_{1m}) =
   \sum_{p=1}^{\infty}\frac{1}{p!}\rho_{\star}^{p-2}Q^{1}_{p,1,m}|_{\mathscr{I}^{+}}.
\end{align}
As before, one has the following :
\begin{IndRemark}\label{rem:l1}
   A direct calculation
using Lemma \ref{Lemma:Sol_Jacobi_and_Logs} gives
  \begin{align}
    Q^{1}_{1;1,m}(\tau) &=-2^{-1}D_{11m}(1-\tau)^{-1},\\
    Q^{1}_{p;1,m}(\tau) &=2^{2-p}(p-1)p(p+1)A_{2,p,m}(1-\tau)^{p-2} \qquad
    \text{for}\qquad p\neq 1.
  \end{align}
\end{IndRemark}
Using Remark \ref{rem:l1} one notices that only the $p=2$ term in
equation \eqref {NP1_precomputation} contributes to the sum.
Direct evaluation then gives
 \begin{align}
\mathcal{N}^{+}_{1,m}=3A_{2,1,m}.
\end{align}

 \subsection{The $i^0$-cylinder logarithmic NP constants}\label{sec:LogarithmicNP_main}

In subsection \ref{sec:ClassicalNP_main} it was shown that if the regularity
condition of Remark \ref{Remark:logfreeRemark} is not satisfied then the
classical NP constants are not well-defined.  Nonetheless, the
existence of analogous constants arising for polyhomogeneous 
expansions for the gravitational field (Weyl scalars) have been found
using the Newman-Penrose formalism ---see \cite{Val98, Val99a}.  Hence,
it is natural to ask whether there exist a choice of $f(\tilde{\rho})$
for which the calculation of the associated modified NP-constants (for
the spin-0 field) does not require imposing the regularity
condition. In this section it is shown that the initial data constants
$D_{p;p,m}$ defining the regularity condition are precisely the
$f(\tilde{\rho})=\tilde{\rho}$-modified NP constants.

\subsubsection{The first few $i^0$-cylinder logarithmic NP constants}

In this subsection the $\ell=0$ and $\ell=1$ modified NP constants
${}^{f}\mathcal{N}^{+}_{\ell,m}$ for the spin-0 field are computed for
$f(\tilde{\rho})=\tilde{\rho}$.
Using Proposition \ref{Prop:NPtoFgauge} to express $L$ 
in terms of derivatives respect to the $F$-coordinates gives
\begin{align}
  \tilde{\rho}L (\phi_{\ell m})= \varkappa^{-1}\bme (\phi_{\ell m}) =
  \varkappa^{-1}\sum_{p=0}^{\infty}\bme(a_{p;\ell,m}(\tau)\rho^p) =
  \varkappa^{-1}\sum_{p=0}^{\infty}Q^{0}_{p;\ell,m}(\tau)\rho^p.
\end{align}
Hence, for $\ell=0$ one has
\begin{align}
  \mathcal{}^{\tilde{\rho}}\mathcal{N}^{+}_{0,0} =
   \lim_{\substack{\rho \to \rho_{\star} \\ \tau \to 1}} \; \varkappa^{-1}\bme(\phi_{00}) =
  \sum_{p=0}^{\infty}\rho_{\star}^p(\varkappa^{-1}Q^{0}_{p;0,0})|_{\mathscr{I}^{+}}.
\end{align}
Using Remark \ref{rem:l0} and evaluating at the critical set $I^{+}$
one gets
\begin{align}
\mathcal{}^{\tilde{\rho}}\mathcal{N}^{+}_{0,0} = \frac{1}{2}D_{0,0,0}.
\end{align}
Similarly, for the $\ell=1$ the relevant quantity to evaluate is
\begin{align}
  \tilde{\rho} L (\bme^{+}) \phi_{1m}= \varkappa^{-1}\bme (\bme^{+})
  \phi_{1m} = 4^{-1}\varkappa^{-1}(\Lambda_{+})^{-2}
  (\bme^{+})^2\phi_{1m}=
  4\varkappa^{-1}(\Lambda_{+})^{2}\sum_{p=1}^{\infty}
  \frac{1}{p!}\rho^pQ^{1}_{p;\ell,m},
\end{align}
where in the last equality we have used  equations \eqref{eq:defQ1} and
\eqref{eq:DefQ1IntermsQ0}.
Hence, for $\ell=1$ one has
\begin{align}
  \mathcal{}^{\tilde{\rho}}\mathcal{N}^{+}_{1,m} =
   \lim_{\substack{\rho \to \rho_{\star} \\ \tau \to 1}} \; \varkappa^{-1}\bme (\bme^{+}) \phi_{1m} =
  \sum_{p=1}^{\infty}\frac{1}{p!}\rho_{\star}^{p-1}(\varkappa^{-1}Q^{1}_{p;1,m})|_{\mathscr{I}^{+}}.
\end{align}
Using Remark \ref{rem:l1} and evaluating at the critical set $I^{+}$
one gets
\begin{align}
\mathcal{}^{\tilde{\rho}}\mathcal{N}^{+}_{1,m} =- \frac{1}{4}D_{1,1,m}.
\end{align}
\subsection{The induction argument}
\label{sec:induction_argument}
In this section the $\ell$-NP constants are computed inductively.
\begin{proposition}\label{lemmaMainNPplus}
\begin{align}\label{eq:bmelplusonephilm}
  (\bme^{+})^{n+1}\phi_{\ell m}=
  (2\Lambda_{+})^{2(n+1)}\sum_{p=\ell}^{\infty}\frac{1}{p!}Q^{n}_{p;\ell,m}(\tau)\rho^{p}
\end{align}
with
\begin{align}\label{eq:Qn}
  Q^{n}_{p;\ell,m}(\tau)= \sum_{q=0}^{n+1}(-1)^q p^{[q]} {n+1 \choose
    q}(1+\tau)^{n+1-q}a_{p;\ell,m}^{(n+1-q)}(\tau)
\end{align}
where $a_{p;\ell,m}^{(n)}(\tau):= d^{n}(a_{p;\ell,m}(\tau))/d\tau^n$
denotes the $n$-th derivative and $p^{[q]}= (p+q-1)!/(p-1)!$ is the
rising factorial.
 \end{proposition}
\begin{proof}
To prove Proposition \ref{lemmaMainNPplus} one proceeds inductively as
follows.  Equations \eqref{eq:bmeplus1philmraw}, \eqref{eq:defQ0}, 
\eqref{eq:bmeplus2philmraw} and \eqref{eq:defQ1}  are recovered when
substituting $\ell=0$ and $\ell=1$ in expressions
\eqref{eq:bmelplusonephilm} and \eqref{eq:Qn}.  This constitutes the
basis of induction.  Assuming that equations
\eqref{eq:bmelplusonephilm} and \eqref{eq:Qn} are valid (induction
hypothesis) and applying $\bme^{+}$ to equation
\eqref{eq:bmelplusonephilm} renders
\begin{align}\label{eq:inductionNp2step1}
  (\bme^+)^{n+2}\phi_{\ell m}= \bme^{+}((\bme^{+})^{n+1}\phi_{\ell
    m})=2^2
  (\Lambda_{+})^{2}\bme\Big((2\Lambda_{+})^{2(n+1)}\sum_{p=\ell}^{\infty}
  \rho^p Q^{n}_{p;\ell,m}(\tau)\Big).
\end{align}
Using equations \eqref{eq:Fframe} and \eqref{eq:LorentzTransf} a direct
calculation gives
\begin{align}\label{eq:bmederLambda}
\bme(\Lambda_{+})^{2n}=-n(\Lambda_{+})^{2n}.
\end{align}
Expanding expression \eqref{eq:inductionNp2step1} assisted with
equation \eqref{eq:bmederLambda} one obtains
\begin{align}\label{eq:inductionNp2step2}
  (\bme^+)^{n+2}\phi_{\ell m}=  
  (2\Lambda_{+})^{2(n+2)}\Big(\sum_{p=\ell}^{\infty} \bme(\rho^p
  Q^{n}_{p;\ell,m}) -(n+1)\rho^{p}Q^{n}_{p;\ell,m}\Big)=
  (2\Lambda_{+})^{2(n+2)}\sum_{p=\ell}^{\infty}\rho^{p}R^{n}_{p;\ell, m}
\end{align}
where
\begin{align}\label{eq:QNrecursive}
  R^{n}_{p;\ell,m}=(1+\tau)\dot{Q}^{n}_{p;\ell,m}-(p+n+1)Q^{n}_{p;\ell,m}.
\end{align}
Substituting equation \eqref{eq:Qn} (induction hypothesis) and
rearranging gives
\begin{align}
  & R^{n}_{p;\ell,m} = (1+\tau)\frac{d}{d\tau}\Bigg[
    \sum_{q=0}^{n+1}(-1)^q p^{[q]} {n+1 \choose
      q}(1+\tau)^{n-q+1}a^{(n-q+1)} \Bigg] \nonumber \\ & \qquad
  \qquad \qquad \qquad \qquad \qquad \qquad \qquad
  -(p+n+1)\sum_{q=0}^{n+1}(-1)^q p^{[q]} {n+1 \choose
    q}(1+\tau)^{n-q+1}a^{(n-q+1)} \nonumber \\ =
  &\sum_{q=0}^{n+1}(-1)^q p^{[q]} {n+1 \choose q}
  (1+\tau)^{n-q+2}a^{(n-q+2)} + \sum_{q=0}^{n+1}(-1)^{q+1}p^{[q]} {n+1
    \choose q} (p+q)(1+\tau)^{n-q+1}a^{(n-q+1)}. \nonumber
\end{align}
In the last expression we have shorten the notation by denoting
$a^{(n)}_{p;\ell,m}(\tau)$ simply as $a^{(n)}$.  Stripping out the
first element in the sum ($q=0$) in the first term and the last
element of the sum ($q=n+1$) in the second term gives
\begin{align}
  R^{n}&_{p;\ell,m} = (1+\tau)^{n+2}a^{(n+2)} + \sum_{q=1}^{n+1}(-1)^q
  p^{[q]} {n+1 \choose q} (1+\tau)^{n-q+2}a^{(n-q+2)} \nonumber \\ & +
  \sum_{q=0}^{n}(-1)^{q+1} p^{[q]}{n+1 \choose
    q}(p+q)(1+\tau)^{n-q+1}a^{(n-q+1)} + (-1)^{n+2}p^{[n+1]}(p+n+1)a.
\end{align}
It follows from the definition of the rising factorial that
$p^{[i+1]}=(p+i)p^{[i]}$, in particular $p^{[n+2]}=(p+n+1)p^{[n+1]}$
. Using the latter fact and relabelling indices so that both sums start at
$q=0$ and end at $q=n$ gives
\begin{align}
  R^{n}_{p;\ell,m} = & (1+\tau)^{n+2}a^{(n+2)} +
  \sum_{q=0}^{n}(-1)^{q+1} p^{[q+1]} {n+1 \choose q+1}
  (1+\tau)^{n-q+1}a^{(n-q+1)} \nonumber \\ & \qquad \qquad +
  \sum_{q=0}^{n}(-1)^{q+1} p^{[q]} {n+1 \choose q}(p+q)
  (1+\tau)^{n-q+1}a^{(n-q+1)} + (-1)^{n+2}p^{[n+2]}a.
\end{align}
Using that $p^{[q+1]}=(p+q)p^{[q]}$ and that
\begin{align}
 {n+1 \choose q+1} + {n+1 \choose q}= {n+2 \choose q+1},
\end{align}
then one gets
\begin{align}
  R^{n}_{p;\ell,m} = & (1+\tau)^{n+2}a^{(n+2)} +
  \sum_{q=0}^{n}(-1)^{q+1} p^{[q+1]} {n+2 \choose q+1}
  (1+\tau)^{n-q+1}a^{(n-q+1)} + (-1)^{n+2}p^{[n+2]}a.
\end{align}
Observing that latter can be rewritten more compactly as
\begin{align}
  R^{n}_{p;\ell,m} = \sum_{q=0}^{n+2}(-1)^{q} p^{[q]} {n+2 \choose q}
  (1+\tau)^{n-q+2}a^{(n-q+2)},
\end{align}
and recalling equation \eqref{eq:Qn} then one concludes that
\begin{align}\label{eq:RnIsQnplusone}
R^{n}_ {p;\ell,m}=Q^{n+1}_{p;\ell,m}.
\end{align}
\end{proof}
Notice that Proposition \ref{lemmaMainNPplus} provides, via equations
 \eqref{eq:QNrecursive} and \eqref{eq:RnIsQnplusone}, with a recursive
way to compute $Q^{n}_{p;\ell,m}$ as recorded in the following:
\begin{corollary}
  $Q^{n}_{p;\ell,m}$ satisfies
  \begin{align}\label{eq:Qnrecursion_coro}
Q^{n+1}_{p;\ell,m}=(1+\tau)\dot{Q}^{n}_{p;\ell,m}
-(p+n+1)Q^{n}_{p;\ell,m}.
  \end{align}
\end{corollary}

\begin{IndRemark}
The rising $x^{[n]}$ and falling $x_{[n]}$ factorials are defined as
\begin{align}
 x^{[n]}:= \prod_{i=1}^{n}(x+i-1), \qquad
 x_{[n]}:=\prod_{i=1}^{n}(x-i+1).
\end{align}
We record here as well that
\begin{subequations}\label{eq:firstDerJac}
\begin{align}
\frac{d}{d\tau}J_{\ell-k}^{(-p+k,p+k)}= \frac{1}{2}(\ell +k
+1)J_{\ell-k-1}^{(-p+k+1,p+k+1)} ,
\end{align}
\begin{align}
\frac{d}{d\tau}J_{\ell-k}^{(p+k,-p+k)}= \frac{1}{2}(\ell +k
+1)J_{\ell-k-1}^{(p+k+1,-p+k+1)}.
\end{align}
\end{subequations}
\end{IndRemark}

In principle, with equation \eqref{eq:Qn} and the solution
$a_{p;\ell,m}(\tau)$ given in Lemma \ref{Lemma:Sol_Jacobi_and_Logs}
one can obtain $Q^{n}_{p;\ell,m}(\tau)$ explicitly for any arbitrary
parameters $n,p,\ell,m$. However, for computing the NP-constants only
the $n=\ell$ case is relevant.  In the following we investigate the
structure of $Q^{\ell}_{p;\ell,m}$ in terms of the initial data
constants $A_{p,\ell,m}$, $B_{p,\ell,m}$, $C_{p,\ell,m}$ and
$D_{p,\ell,m}$.

\begin{proposition}
For $p \neq \ell$
  \begin{align}\label{eq:ind_Qpnotl}
    Q^{n}_{p;\ell,m}(\tau) =
    (\ell+1)^{[n+1]}\Big(\frac{1+\tau}{2}\Big)^{p+n+1}
    B_{p,\ell,m}J_{\ell-n-1}^{(-p+n+1,p+n+1)}(\tau)+\sum_{k=0}^{n+1}
    A_{p,\ell,m}h^{n}_{k}(\tau)J_{\ell-k}^{(p+k,-p+k)}(\tau)
  \end{align}
  where $h^{n}_{k}(\tau)$ are polynomials in $\tau$ which can be
  determined recursively as follows:
   \begin{align}
h_{k}^{n+1}:= \Big(\frac{1+\tau}{2}\Big)(\ell+ k+ 1)h_{k-1}^{n} +
(1+\tau)\dot{h}_{k}^{n}-(p+n+1)h_{k}^{n},
   \end{align}
   starting with
   \begin{align}
     h_{0}^{0}(\tau) =-p \Big(\frac{1-\tau}{2}\Big)^{p-1}, \qquad
     h_{1}^0(\tau) =(\ell+1) \Big(\frac{1-\tau}{2}\Big)^{p}
     \Big(\frac{1+\tau}{2}\Big),\qquad h_{-1}^{n}=h^{n}_{n+2}=0.
  \end{align}
\end{proposition}

\begin{proof}
  One proceeds inductively as follows. A direct calculation using
  Lemma \ref{Lemma:Sol_Jacobi_and_Logs} for the case $p \neq \ell$, gives
\begin{align}
  Q^{0}_{p;\ell,m}= (1+\tau)\dot{a}_{p;\ell,m} = B_{p,\ell,m}(\ell+1)
  \Big(\frac{1+\tau}{2}\Big)^{p+1}J_{\ell-1}^{(-p+1,p+1)}-A_{p,\ell,m}p
  \Big(\frac{1-\tau}{2}\Big)^{p-1} J_{\ell}^{(p,-p)}\nonumber
  \\ +(\ell+1)A_{p,\ell,m} \Big(\frac{1-\tau}{2}\Big)^{p}
  \Big(\frac{1+\tau}{2}\Big)J_{\ell-1}^{(p+1,-p+1)}.
\end{align}
This proves the case $n=0$ with
\begin{align}\label{eq:h0andh1}
  h_{0}^{0}(\tau)=-p \Big(\frac{1-\tau}{2}\Big)^{p-1}, \qquad
  h_{1}^0(\tau)=(\ell+1) \Big(\frac{1-\tau}{2}\Big)^{p}
  \Big(\frac{1+\tau}{2}\Big),
\end{align}
which constitutes the basis of induction.  Now, assume equation
\eqref{eq:ind_Qpnotl} is valid as the induction hypothesis.  To prove
the induction step we need to compute $Q^{n+1}_{p;\ell,m}$. Using
equation \eqref{eq:Qnrecursion_coro} gives
\begin{align}
  Q^{n+1}_{p;\ell,m}=\sum_{k=0}^{n+1}A_{p,\ell,m} \Big[(1+\tau)\Big(
    h^{n}_{k} \dot{J}_{\ell-k}^{(p+k,-p+k)} +
    J_{\ell-k}^{(p+k,-p+k)}\dot{h}^{n}_{k} \Big) -(p+n+1)h^{n}_{k}
    J_{\ell-k}^{(p+k,-p+k)}\Big] \nonumber \\ +2 B_{p,\ell,m}
  (\ell+1)^{[n+1]}
  \Big(\frac{1+\tau}{2}\Big)^{p+n+2}\dot{J}_{\ell-n-1}^{(-p+n+1,p+n+1)}.
\end{align}
Notice from the previous expression that, in those terms with coefficient
$B$ only $\dot{J}$ appears as the terms of the form $B J$ cancel
out.

\begin{IndRemark}
  As we will see in Corollary \ref{coro:propNPplusQ}, this
  cancellation is, ultimately, the responsible for the classical NP
  constants at $\mathscr{I}^{+}$ to be determined \emph{only by
  $A_{p;\ell,m}$} and not a combination of $A_{p;\ell,m}$ and
  $B_{p;\ell,m}$  as one could expect in general.  
\end{IndRemark}

Using identity \eqref{eq:firstDerJac} to compute the derivatives
of the Jacobi polynomials and using that
$(\ell+n+2)(\ell+1)^{[n+1]}=(\ell+1)^{[n+2]}$ gives
\begin{align}
  Q^{n+1}_{p;\ell,m}=\sum_{k=0}^{n+1}A_{p,\ell,m} \Big[
    \Big(\frac{1+\tau}{2}\Big)(\ell+k+1)h^{n}_{k} J_{\ell-k-1}^{(p+k+1,-p+k+1)}
    +\Big((1+\tau) \dot{h}^{n}_{k} -(p+n+1)h^{n}_{k}\Big)
    J_{\ell-k}^{(p+k,-p+k)} \Big] \nonumber \\ + B_{p,\ell,m}
  (\ell+1)^{[n+2]} \Big(\frac{1+\tau}{2}\Big)^{p+n+2}
  J_{\ell-n-2}^{(-p+n+2,p+n+2)}.
\end{align}
Rearranging the sum, the latter expression can be rewritten as
\begin{align}
    Q^{n+1}_{p;\ell,m}(\tau) =
    (\ell+1)^{[n+2]}\Big(\frac{1+\tau}{2}\Big)^{p+n+2}
    B_{p,\ell,m}J_{\ell-n-2}^{(-p+n+2,p+n+2)}(\tau)+\sum_{k=0}^{n+2}
    A_{p,\ell,m}h^{n+1}_{k}(\tau)J_{\ell-k}^{(p+k,-p+k)}(\tau)
  \end{align}
   where
  \begin{align}
h_{k}^{n+1}:= \Big(\frac{1+\tau}{2}\Big)(\ell+ k+ 1)h_{k-1}^{n} +
(1+\tau)\dot{h}_{k}^{n}-(p+n+1)h_{k}^{n} 
  \end{align}
  with $h^n_{-1}=h^{n}_{n+2}=0$ and the iteration starts with
  $h_{0}^{0}$ and $h_{1}^{0}$ as given in equation \eqref{eq:h0andh1}.
\end{proof}

\begin{corollary}\label{coro:propNPplusQ}
  For $p\neq l$ and $n=l$ one has
    \begin{align}\label{eq:ind_Qpnotl_pequaln}
    Q^{n}_{p;n,m}(\tau) =\sum_{k=0}^{n+1}
    A_{p,n,m}h^{n}_{k}(\tau)J_{n-k}^{(p+k,-p+k)}(\tau).
  \end{align}
\end{corollary}
\begin{proof}
  This follows by noticing that the Jacobi polynomial multiplying the
  term $B_{p,l,m}$ in equation \eqref{eq:ind_Qpnotl}
vanishes if $\ell=n$.
\end{proof}

\begin{proposition}\label{prop:pequallQn}
  \begin{align}\label{eq:prop:pequallQn}
Q^{n}_{p;p;m}(\tau)=g^n_{p}(\tau)D_{p,p,m} +
(-1)^{n+1}p_{[n+1]}\Big(\frac{1+\tau}{2}\Big)^{p+n+1}\Big(\frac{1-\tau}{2}\Big)^{p-n-1}\Bigg(C_{p,p,m}
+D_{p,p,m} \int_{0}^{\tau} \frac{ds}{(1-s^2)^{p+1}}\Bigg)
  \end{align}
  where $g^n_p(\tau)$ is a polynomial in $\tau$ determined recursively
  by
\begin{align}
  g^{n+1}_p &
  =(1+\tau)\dot{g}^{n}_p+(-1)^{n+1}2^{-2p}p_{[n+1]}(1+\tau)^{n+1}(1-\tau)^{-n-2}-(p+n+1)g^{n}_p,
  \\ g^{0}_p & =2^{-2p}(1-\tau)^{-1}.
\end{align}

\begin{proof}
  To simplify the notation let
  $F(\tau):=\int_{0}^{\tau}(1-s^2)^{-p-1}ds$, so that
  $\dot{F}=(1-\tau)^{-p-1}(1+\tau)^{-p-1}$. Using this notation, a
  direct calculation using Lemma \ref{Lemma:Sol_Jacobi_and_Logs} for the
  case $p = \ell$, gives
\begin{align}
Q^{0}_{p;p,m}=(1+\tau)\dot{a}_{p,p,m}=2^{-2p}D_{p,p,m}(1-\tau)^{-1} -
p\Big(\frac{1+\tau}{2}\Big)^{p+1}\Big(\frac{1-\tau}{2}\Big)^{p-1}(C_{p,p,m}
+D_{p,p,m}F).
\end{align}
This shows the $n=0$-case of equation \eqref{eq:prop:pequallQn} with
$g^0_p=2^{-2p}(1-\tau)^{-1}$.  The general case is proven by induction
as follows. Assume, as induction hypothesis, that
\eqref{eq:prop:pequallQn} holds. To compute $Q^{n+1}_{p;p,m}$ we
will use equation \eqref{eq:Qnrecursion_coro}.
To do this calculation in steps, first observe that
\begin{align}
\dot{Q}^{n}_{p,p,m} =
D_{p,p,m}G^{n}_p+(-1)^{n+1}p_{[n+1]}2^{-2p}(C_{p,p,m}
+D_{p,p,m}F)(1+\tau)^{p+n}(1-\tau)^{p-n-2}2(n+1-p\tau)
\end{align}
where $
G^{n}_p:=\dot{g}^{n}_p+(-1)^{n+1}2^{-2p}p_{[n+1]}(1+\tau)^{n}(1-\tau)^{-n-2}$
assisted with equation \eqref{eq:Qnrecursion_coro} a calculation
renders
\begin{align}
  Q^{n+1}_{p,p,m} & =\Big((1+\tau)G^{n}_p-(p+n+1)g^{n}_p\Big)D_{p,p,m} +
  \nonumber \\ & \qquad \quad (-1)^{n+1}p_{[n+1]}2^{-2p}(C_{p,p,m}
  +D_{p,p,m}F)(1+\tau)^{p+n+2}(1-\tau)^{p-n-2}(-p+n+1).
\end{align}
Using that $(-p+n+1)p_{[n+1]}=-p_{[n+2]}$ one finally obtains
\begin{align}
  Q^{n+1}_{p,p,m} =g^{n+1}_p D_{p,p,m} +
  (-1)^{n+2}p_{[n+2]}\Big(\frac{1+\tau}{2}\Big)^{p+n+2}
  \Big(\frac{1-\tau}{2}\Big)^{p-n-2}\Bigg(C_{p,p,m} +D_{p,p,m} F\Bigg)
\end{align}
with
\begin{align}
g^{n+1}_p=(1+\tau)\dot{g}^{n}_p+(-1)^{n+1}2^{-2p}p_{[n+1]}(1+\tau)^{n+1}(1-\tau)^{-n-2}-(p+n+1)g^{n}_p.
\end{align}
\end{proof}
\end{proposition}

\newpage

\begin{corollary}\label{coro:propNPplusQlog}
  For $p=l=n$ one has
    \begin{align}\label{eq:ind_Qpequall_pequaln}
    Q^{n}_{n;n,m}(\tau) =g^n_n(\tau)D_{n,n,m}.
  \end{align}
\end{corollary}
\begin{proof}
  This follows from equation \eqref{eq:prop:pequallQn} by noticing
  that $p_{[n+1]}$ vanishes if $p=n$.
\end{proof}

\begin{IndRemark}\label{Remark:Qngen}
   Although $h_{k}^{n}(\tau)$ and $g^{n}_{p}(\tau)$ can be
  computed recursively, as shown above, solving this recursion
  relation to obtain an explicit expression for any $(n,p)$ is not an
  easy task.  Nonetheless, assuming the NP constants are well-defined
  (i.e. assuming that the limit defining the NP constants does exist), for
  the purposes of identifying which parameters in the initial data (or
  rather, which ones do not) determine the value of the NP constants,
  Corollaries \ref{coro:propNPplusQ} and \ref{coro:propNPplusQlog} are
  enough. Furthermore, since the formulae derived in this section are
  written in terms of special functions and recursive relations one
  can easily compute $Q^{n}_{p;n,m}$ explicitly for arbitrary large $n$
  assisted with computer algebra programs such as {\tt{Mathematica}}. Some
  experimentation reveals that in fact one expects that for general
  $(n=\ell,p)$:
  \begin{subequations}\label{rem:n}
  \begin{align}
    Q^{n}_{n;n,m}(\tau) &=\alpha_{n}D_{n,n,m}(1-\tau)^{-1},
    \\ Q^{n}_{p;n,m}(\tau) &=\beta_{n,p}\prod_{i=0}^{2n} (p+n-i)
    A_{n,p,m}(1-\tau)^{p-n-1} \qquad \text{for}\qquad p\neq n,
  \end{align}
  \end{subequations}
  where $\alpha_{n}$ and $\beta_{n,p}$ are some constants.  
\end{IndRemark}

\subsection{The NP constants in terms of initial data}

In this subsection the results of the induction arguments of
subsection \ref{sec:induction_argument} are put together to assess how
the NP constants at $\mathscr{I}^{\pm}$ are determined in terms of the
initial data parameters $A_{p,\ell,m}$, $B_{p,\ell,m}$,
$C_{p,\ell,m}$, and $D_{p,\ell,m}$.

\begin{proposition}
Let $\tilde{\phi}$ be a spin-0 field propagating in Minkowski spacetime
with initial data close to $i^0$ as encoded in equation \eqref{eq:ID_field}
for which the  regularity condition is satisfied $D_{p,p,m}=0$.
If the classical
NP constants at $\mathscr{I}^+$ are well-defined then
 \begin{align}
   \mathcal{N}^{+}_{\ell,m}=q^{+}(\ell) \;A_{\ell+1,\ell,m},
 \end{align}
 where $q^{+}(\ell)$ is some numerical factor.
 In other words, the
  $\mathscr{I}^+$ classical NP constants
 depend only on the initial data parameter $A_{\ell+1,\ell,m}$.
\end{proposition}
   
\begin{proof}
Using equations \eqref {eq:classicalNP} and \eqref{eq:bmelplusonephilm} gives
  \begin{align}
\mathcal{N}^{+}_{\ell,m}= (\bme^{+})^{\ell+1}\phi_{\ell m}|_{\mathcal{C}^{+}} =
(2\Lambda_{+})^{2(\ell+1)}\sum_{p=\ell}^{\infty}\frac{1}{p!}
Q^{\ell}_{p;\ell,m}(\tau)\rho^{p}|_{\mathcal{C}^{+}}
=\sum_{p=\ell}^{\infty}\frac{1}{p!}\rho_{\star}^{p-(\ell+1)}Q^{\ell}_{p;\ell,m}|_{\mathscr{I}^{+}}.
  \end{align}
  Imposing the regularity condition and assuming the NP constants
  are well-defined, evaluation at the critical set $I^{+}$ gives 
 \begin{align}
\mathcal{N}^{+}_{\ell,m}
=\frac{1}{(\ell+1)!}Q^{\ell}_{\ell+1;\ell,m}|_{\mathscr{I}^{+}}.
 \end{align}
 Inspection of $Q^{\ell}_{\ell+1;\ell,m}|_{\mathscr{I}^{+}}$ using
 Corollary \ref{coro:propNPplusQ} then renders
  \begin{align}
   \mathcal{N}^{+}_{\ell,m}=q^{+}(\ell) \;A_{\ell+1, \ell,m},
  \end{align}
  where $q^{+}(\ell)$ is a numerical factor.
\end{proof}

\begin{proposition}
Let $\tilde{\phi}$ be a spin-0 field propagating in Minkowski
spacetime with initial data close to $i^0$
as encoded in equation \eqref{eq:ID_field}.  If the
$f(\tilde{\rho})=\tilde{\rho}$ modified NP constants at
$\mathscr{I}^+$ are well-defined then
\begin{align}
   {}^{\tilde{\rho}}\mathcal{N}^{+}_{\ell,m}=q^{+}(\ell) \;D_{\ell,\ell,m}
 \end{align}
 where $q^{+}(\ell)$ is some numerical factor.  In other words, the
 $i^0$-cylinder logarithmic NP constants at $\mathscr{I}^{+}$
 depend only on the initial
 data parameter $D_{\ell,\ell,m}$.
\end{proposition}

\begin{proof}             
Using equations \eqref{eq:DefModifiedNP} and
\eqref{eq:bmelplusonephilm} gives
  \begin{align}
   {}^{\tilde{\rho}}\mathcal{N}^{+}_{\ell,m}= \tilde{\rho}L
   (\bme^{+})^{\ell}\phi_{\ell m}|_{\mathcal{C}^{+}}=
   \varkappa^{-1}(2\Lambda_+)^{-2}(\bme^{+})^{\ell+1}\phi_{\ell
     m}|_{\mathcal{C}^{+}}
   =\sum_{p=\ell}^{\infty}\frac{1}{p!}\rho_{\star}^{p-\ell}
   (\varkappa^{-1}Q^{\ell}_{p;\ell,m})|_{\mathscr{I}^{+}}.
  \end{align}
  Assuming that the $i^0$-cylinder logarithmic NP constants are well-defined,
  evaluation at the critical set $I^{+}$ gives
  \begin{align}
   {}^{\tilde{\rho}}\mathcal{N}^{+}_{\ell,m}
   =\frac{1}{\ell!} (\varkappa^{-1}Q^{\ell}_{\ell;\ell,m})|_{\mathscr{I}^{+}}
  \end{align}
  Inspection of $Q^{\ell}_{\ell;\ell,m}|_{\mathscr{I}^{+}}$ using
  Corollary \ref{eq:ind_Qpequall_pequaln} then renders
   \begin{align}
   {}^{\tilde{\rho}}\mathcal{N}^{+}_{\ell,m}=q^{+}(\ell) \;D_{\ell,\ell,m}
 \end{align}
 where $q^{+}(\ell)$ is some numerical factor.
\end{proof}

With minimal changes, using Proposition \ref{Prop:NPtoFgauge}, and
equation \eqref{eq:classicalNPMinus} the classical NP constants can be
computed at $\mathscr{I}^{-}$. Analogous results of those
of subsections \ref{sec:ClassicalNP_main}, \ref{sec:LogarithmicNP_main}
and \ref{sec:induction_argument} can be derived and
\emph{mutatis mutandis}, the time-dual propositions are:

\begin{proposition}
Let $\tilde{\phi}$ be a spin-0 field propagating in Minkowski
spacetime with initial data close to $i^0$ as encoded in
equation \eqref{eq:ID_field} for which the regularity condition is
satisfied $D_{p,p,m}=0$.  If the classical NP constants at
$\mathscr{I}^-$ are well-defined then
 \begin{align}
   \mathcal{N}^{-}_{\ell,m}=q^{-}(\ell) \;B_{\ell+1,\ell,m}
 \end{align}
 where $q^{+}(\ell)$ is some numerical factor.  In other words, \emph{
 $\mathscr{I}^-$ classical NP constants} depend only on the initial
 data parameter $B_{\ell+1,\ell,m}$.
\end{proposition}

\begin{proposition}
Let $\tilde{\phi}$ be a spin-0 field propagating in Minkowski
spacetime with initial data close to $i^0$
as in equation \eqref{eq:ID_field}.  If the
$f(\tilde{\rho})=\tilde{\rho}$ modified NP constants at
$\mathscr{I}^-$ are well-defined then
\begin{align}
   {}^{\tilde{\rho}}\mathcal{N}^{-}_{\ell,m}=q^{-}(\ell) \;D_{\ell,\ell,m}
 \end{align}
 where $q^{+}(\ell)$ is some numerical factor.  In other words, the
$i^0$-cylinder NP constants at $\mathscr{I}^{-}$ depend only on the initial
 data parameter $D_{\ell,\ell,m}$.
\end{proposition}

Comparing the results at $\mathscr{I}^{\pm}$
one can summarise the discussion of this
section in the form of the following 

\begin{corollary}
  Let $\tilde{\phi}$ be a spin-0 field propagating in Minkowski
  spacetime with initial data 
  close to $i^0$ as encoded in equation \eqref{eq:ID_field}:
  
  \begin{itemize}
  \item  If the regularity condition is not satisfied then the classical
    NP constants are not well-defined.
   \item  If the regularity condition is satisfied and
    the classical NP constants are well-defined, then
    the classical NP constants at $\mathscr{I}^{\pm}$
    generically do not coincide. 
  \item If the logarithmic NP constants at $\mathscr{I}^{\pm}$ are well-defined
    then they arise from
    the same part of the initial data and hence, up to a numerical constant,
     coincide.
  \end{itemize}
  \end{corollary}

\section{Higher-order asymptotic expansions close to $i^0$ and $\mathscr{I}^{+}$}

In~\cite{DuaFenGas22} the peeling behaviour of the Weyl tensor was
studied exploiting an asymptotic expansion introduced
in~\cite{DuaFenGasHil21}. The method introduced
in~\cite{DuaFenGasHil21} is based on a heuristic generalisation of
H\"ormander's asymptotic system ~\cite{LinRod03, GasHil18} and was
developed for a formulation of the Einstein field equations in
generalised harmonic gauge designed for hyperboloidal numerical
evolutions ---see \cite{DuaFenGasHil22a}.  Interestingly, the
polyhomogeneous expansions for the Weyl tensor of \cite{DuaFenGas22}
look similar to those obtained using the conformal Einstein field
equations in \cite{Fri98a, Fri03, Val07, GasVal18, Fri18}.  However,
as shown in \cite{DuaFenGasHil22b}, the expansions of
\cite{DuaFenGas22} cannot encode the logs related to the critical sets
$I^{\pm}$---which will be called $i^0$-cylinder-logs for short. To see
this, it is enough to examine (the simplest scenario) the solutions to
the wave equation in flat spacetime. In \cite{DuaFenGasHil21} the
expansion for the solution for the source-free wave equation (or
good-field in the naming conventions of \cite{DuaFenGasHil21}) is
log-free while using the $i^0$-cylinder framework one sees this is not
the case unless the initial data is fine-tuned (regularity condition)
---see in Lemma \ref{Lemma:Sol_Jacobi_and_Logs} and Remark
\ref{Remark:logfreeRemark}.  Nonetheless, in \cite{DuaFenGasHil22b} it
was discussed how with a more careful reading of H\"ormander's
asymptotic system (leading order asymptotic system) one can recover
the ``leading log'' in the physical field and determine a condition
(at the level of initial data) that controls the appearance of the
leading log. This condition corresponds, in the unphysical set-up, to
choosing initial data so that $D_{0,0,0}=0$ and, in the
physical-asymptotic-system approach, to the condition that
\begin{align}\label{no-leading-log-physicalgood}
\partial_{\tilde{v}}(\tilde{\rho}\tilde{\phi}) \simeq
\mathcal{O}(\tilde{\rho}^{-2}) \qquad \text{on}\qquad \mathcal{S}.
\end{align}
However, it is clear that even with this improvement, the
asymptotic-system approach discussed in section 3.5 (revisiting the
asymptotic system) of \cite{DuaFenGasHil22b} cannot recover the
\emph{higher-order $i^0$-cylinder-logs} controlled by the coefficient
$D_{n,n,m}$ with $n \geq 1$. This is because the appearance of logs
depend not only on the rough decay of the field at infinity ---see
Proposition 3 in \cite{DuaFenGasHil22b}--- but on the initial data for
the highest spherical harmonic mode at fixed order in $\rho$.  Since
the asymptotic system heuristics and its current higher-order version
in \cite{DuaFenGasHil21, DuaFenGas22, DuaFenGasHil22a} do not make
such harmonic decomposition, those asymptotic expansions are
insensitive to the presence of the higher-order
$i^0$-cylinder-logs. In this section it is shown how the identities in
equation \eqref{eq:cons_laws}---see also Appendix \ref{App:A}---
provide a potential path for a modification of the higher-order
asymptotic system approach of \cite{DuaFenGasHil21, DuaFenGas22,
  DuaFenGasHil22a} which is sensitive to the $i^0$-cylinder-logs, at
least, in flat spacetime.  Whether this heuristic method can be
extended for asymptotically flat spacetimes ---in a spirit similar to
that of \cite{DuaFenGasHil21}--- is left for future work.

\medskip

To start the discussion, observe from equation \eqref{eq:exp_phi_lm}
that $\phi_{\ell m}=\phi_{\ell m}(\tau,\rho)$.  Using that $\phi=
\Theta \tilde{\phi}$, then the expansion \eqref{eq:exp_phi_lm} can be
thought as an expansion in the physical set-up via
\begin{align}\label{eq:readPhysFieldFromUnphys}
  \tilde{\phi}_{\ell m} = \tilde{\rho}^{-1}\phi_{\ell m}
  (\tau(\tilde{t},\tilde{\rho}),\rho(\tilde{t},\tilde{\rho})).
\end{align}
In other words, one can implicitly infer an expansion for the physical
field in terms of its unphysical version. Notice that for the source-free
wave equation, Corollary
\ref{coro:main_commutation} for $\ell=0$ implies
\begin{align}
\partial_{\tilde{u}} \partial_{\tilde{v}} \phi_{00}=0,
\end{align}
which can be thought as the \emph{leading order asymptotic system}.
Integrating once renders
\begin{align}
\partial_{\tilde{v}} \phi_{00}=f_{00}(\tilde{v}),
\end{align}
where
\begin{align}
  f_{00} (\tilde{v})=
  \partial_{\tilde{v}}\phi_{00}|_{\tilde{u}=\tilde{u}_{\star}}.
\end{align}
The last equation expressed in terms of the physical field reads $
f_{00}(\tilde{v})=\partial_{\tilde{v}}(\tilde{\rho}\tilde{\phi}_{00})|_{\tilde{u}=\tilde{u}_{\star}}$
---compare with equation \eqref{no-leading-log-physicalgood}.
Integrating again and using equation
\eqref{eq:readPhysFieldFromUnphys} gives the following expression for
the leading order physical field
\begin{align}
\tilde{\phi}_{00}= \frac{1}{\tilde{\rho}}(F_{00}(\tilde{v})
+G_{00}(\tilde{u}))
\end{align}
where $G_{00} = \phi_{00}|_{\tilde{v}=\tilde{v}_{\star}}$ and $F_{00}=
\int_{\tilde{v}_{\star}}^{\tilde{v}}f_{00}(s)ds$.  This discussion
recovers that of \cite{DuaFenGasHil22b} however, it does that in a way
in which it can be generalised to higher-$\ell$. In other words, the
infinite hierarchy of conservation laws of Corollaries
\ref{coro:main_commutation} and \ref{coro:main_commutation_minus} can
be regarded as the \emph{($i^0$-cylinder-log sensitive) higher-order
  asymptotic system}. To see this is the case,  notice that for the
$\mathscr{I}^{+}$-adapted set-up, for the source-free wave equation one has:
\begin{align}
\uL (\tilde{\rho}^{-2\ell} L (\bme^{+})^{\ell}\phi_{\ell m}) =0.
\end{align}
Integrating once gives 
\begin{align}
\tilde{\rho}^{-2\ell}L (\bme^{+})^{\ell}\phi_{\ell m} = f_{\ell m} (\tilde{v}),
\end{align}
which can be written more conveniently as
\begin{align}\label{eq:high-order-asymptotic-system-good}
\tilde{\rho}^{-2\ell-2} (\bme^{+})^{\ell+1}\phi_{\ell m} = f_{\ell m}
(\tilde{v}).
\end{align}
To see that this notion of higher order asymptotic system
is indeed sensitive to
all (higher-order) $i^0$-cylinder-logs it is enough to compute
$f_{\ell m} (\tilde{v})$ for the exact solution (arising from analytic
initial data in a neighbourhood of $i^0$) given in Lemma
\ref{Lemma:Sol_Jacobi_and_Logs}.  For $\ell=0$, expressing equation
\eqref{eq:high-order-asymptotic-system-good} in terms of the
$F$-frame and using Proposition \ref{Prop:NPtoFgauge} gives
\begin{align}
  f_{00}= L \phi_{00}= \Theta^2 \bme^{+} (\phi_{00})=
  2^2 \sum_{p=0}^{\infty}\frac{1}{p!}(1-\tau)^2Q {}^{0}_{p;0,0}(\tau)\rho^{p+1}.
\end{align}
Using equations \eqref{eq:rem:l0} and \eqref{eq:UnphysPhysAdvRet}
one obtains,
\begin{align}\label{eq:f00}
  f_{00}= \frac{2^2D_{000}}{\tilde{v}} - \sum_{p=1}^{\infty}
  \frac{ 2^{2-p}}{(p-1)!}\frac{A_{p,0,m}}{\tilde{v}^{p+1}}.
\end{align}
Hence, the condition $f_{00}(\tilde{v})=\mathcal{O}(\tilde{v}^{-2})$
coincides with expression \eqref{no-leading-log-physicalgood} on $\mathcal{S}$.
This condition, in turn, implies that
$D_{0,0,0}=0$, which is the regularity condition at order $p=0$.  Moreover,
for general $\ell$, using Proposition \ref{lemmaMainNPplus} gives
\begin{align}
  f_{\ell m} = ( 2\Theta \Lambda_{+})^{2(\ell +1)}
  \sum_{p=\ell}^{\infty}\frac{1}{p!}{}Q
      {}^{\ell}_{p;\ell,m}(\tau)\rho^{p} = 2^{2(\ell+1)} 
      \sum_{p=\ell}^{\infty}\frac{1}{p!}{}(1-\tau)^{2(\ell+1)}Q
      {}^{\ell}_{p;\ell,m}(\tau)\rho^{p+\ell +1}.
\end{align}      
One can compute $f_{\ell m}$ explicitly using equation \eqref{eq:Qn}
and Lemma \ref{Lemma:Sol_Jacobi_and_Logs} for fixed $\ell$ to see
that, just as before, $f_{\ell m}$ is given by a power expansion in
$\tilde{v}$ where the slowest decaying term is that corresponding to
$D_{p,p,m}$.  Moreover, using equation \eqref{rem:n} ---see Remark
\ref{Remark:Qngen}--- one realises that the $f_{\ell m}$ associated to the
solution of Lemma \ref{Lemma:Sol_Jacobi_and_Logs} is of the form:
\begin{align}\label{eq:fvAnsatz}
  f_{\ell m}(\tilde{v})= \frac{\alpha_{\ell}D_{\ell, \ell, m
  }}{\tilde{v}^{2\ell +1}} + \sum_{q=0}^{\infty}\beta_{q}
  \frac{A_{\ell +1 +q, \ell, m}}{\tilde{v}^{2(\ell+1)+q}},
\end{align}
for some constants $\alpha_\ell$ and $\beta_{q}$.  Hence, the higher
order no-log condition $D_{\ell, \ell, m }=0$, is equivalent to
require that
\begin{align}
f_{\ell m}(\tilde{v}) \simeq \mathcal{O}(\tilde{v}^{-2(\ell+1)}).
\end{align}
Moreover, one can use the latter to obtain, heuristically, an
asymptotic expansion for the physical field $\tilde{\phi}$ by
integrating equation \eqref{eq:high-order-asymptotic-system-good} for
any assumed form of $f_{\ell m}(\tilde{v})$.  Although not exactly the
same, the choice of $f_{\ell m}(\tilde{v})$ can be thought as the
asymptotic-system counterpart of the choice of Ansatz for $\phi$
---say, as that of equation \eqref{eq:Ansatz}. See Remark
\ref{Remark:AnsatzAndF} for a discussion on the limits of this
analogy. It follows from equation
\eqref{eq:high-order-asymptotic-system-good}
that to obtain $\phi_{\ell m}$ (for any given $f_{\ell m}
(\tilde{v})$) one has to formally integrate $\ell+1$ times the
equation:
\begin{align}\label{eq:high-order-asymptsys-NP}
 (\bme^{+})^{\ell+1}\phi_{\ell m} = \tilde{\rho}^{2(\ell+1)}f_{\ell m}
  (\tilde{v}).
\end{align}

One can do this, in the same spirit of \cite{GasHil18,DuaFenGasHil21},
by considering the field only on outgoing null curves
$\bm\gamma=(\tilde{u}_{\star}, \tilde{v}(\lambda), \theta^A_{\star})$
with $\frac{d\tilde{v}}{d\lambda}=
\frac{1}{4}(\tilde{v}-\tilde{u}_{\star})^{2}$, so that along $\gamma$
one has
\begin{align}
\bme^{+}\phi_{\ell m}(\lambda)= \frac{d}{d\lambda}\phi_{\ell
  m}(\lambda).
\end{align}
Direct integration of $\tilde{v}(\lambda)$ gives
\begin{align}
\tilde{v}=\tilde{u}_{\star} - \frac{4c_{\star}}{2+c_{\star}\lambda},
\end{align}
with $c_{\star}=\sqrt{\tilde{v}'(\lambda_{\star})}$ where ${}'$ denotes $d/d\lambda$.
Hence $\lambda
\sim \tilde{v}^{-1}$ thus along $\bm\gamma$ one can write
\begin{align}\label{eq:high-order-asymptsys-NP-alonggammaGeneral}
  \Big(\frac{d}{d\lambda}\Big)^{\ell+1}\phi_{\ell m}(\lambda) \simeq
  \lambda^{-2(\ell+1)}f_{\ell m} (\lambda),
\end{align}
which can be formally integrated once an specific form of $f_{\ell m}$
is assumed.  In the case of equation \eqref{eq:fvAnsatz} one has
\begin{align}\label{eq:high-order-asymptsys-NP-alonggamma}
  \Big(\frac{d}{d\lambda}\Big)^{\ell+1}\phi_{\ell m}(\lambda) \simeq
  \frac{\alpha_{\ell}D_{\ell \ell m }}{\lambda} +
  \sum_{q=0}^{\infty}\beta_{q} A_{\ell +1 +q, \ell, m}\lambda^q.
\end{align}
Integration shows that the term with $D_{\ell, \ell, m}$ will
generate log-terms.  Integrating $\ell+1$ times and recalling that
$\tilde{\phi}=\Theta \phi$  gives the following expansion for the
physical field
\begin{align}
  \tilde{\phi}_{\ell m}(\lambda) \simeq \frac{1}{\tilde{\rho}} \Bigg(
  \alpha_{\ell}D_{\ell \ell m}\lambda^\ell \ln |\lambda| +
  \sum_{q=0}^{\infty}\beta_{q} A_{\ell +1 +q, \ell,
    m}\lambda^{q+1}\Bigg).
\end{align}
Recalling that $\lambda \sim \tilde{v}^{-1}$ and that $\tilde{\rho}
\sim \tilde{v}$ along $\gamma$ one concludes that
\begin{align}\label{eq:asympt_exp_good}
  \tilde{\phi}_{\ell m} \simeq \;\; \alpha_{\ell}D_{\ell \ell m}
  \frac{\ln |\tilde{\rho}|}{\tilde{\rho}^{\ell + 1}} +
  \sum_{q=0}^{\infty}\frac{\beta_{q} A_{\ell +1 +q, \ell,
      m}}{\tilde{\rho}^{q+1}}.
\end{align}
Observe that the logarithmic terms appear linearly and this is
consistent with translating directly the unphysical solution to
the physical set-up and evaluating at $\tilde{u}=\tilde{u}_{\star}$
as done in \cite{DuaFenGasHil22b}.

\begin{IndRemark}\label{Remark:AnsatzAndF}
As pointed out before, although specifying a form functional form for
$f_{\ell m} (\tilde{v})$ can be thought as the asymptotic-system
analogue of the Ansatz \eqref{eq:Ansatz}, notice that with the 
method of \cite{MinMacVal22} leading to Lemma
\ref{Lemma:Sol_Jacobi_and_Logs} one obtains the solution in a
spacetime neighbourhood of $i^0$ and not only restricted to the curves
$\bm\gamma$ as in the asymptotic system approach. In other words, despite
that equation \eqref{eq:asympt_exp_good} recovers the behaviour of the
solution \emph{towards} $\mathscr{I}^{+}$ as parametrised by $\lambda
\sim \tilde{v}^{-1} \sim \tilde{\rho}^{-1}$, the behaviour
\emph{along} $\mathscr{I}^{+}$ (say the $\tilde{u}$ dependence) of the
solution is lost.
\end{IndRemark}
With the caveat of Remark \ref{Remark:AnsatzAndF} in mind, it is
nonetheless interesting to note that changing the ``$f_{\ell
  m}$-Ansatz'' to
\begin{align}
  f_{\ell m}(\tilde{v})= \frac{ (\ln \tilde{v})^n}{\tilde{v}^{2\ell
      +1}} \qquad \text{for} \qquad n \geq 1,
\end{align}
and applying the asymptotic system heuristics described before,
integrating $\ell+1$-times equation
\eqref{eq:high-order-asymptsys-NP-alonggammaGeneral},
would lead to an expansion
of the form
\begin{align}\label{eq:logToPower}
\tilde{\phi}_{\ell m} \simeq \frac{1}{\tilde{\rho}^{\ell+1}}
\sum_{i=0}^{n}\alpha_i (\ln \tilde{\rho})^i,
\end{align}
where, in contrast to expansion \eqref{eq:asympt_exp_good}, the logs
enter non-linearly.  It would be interesting to see whether such type
of solution can be obtained by modifying the Ansatz
\eqref{eq:Ansatz} appropriately and hence obtaining an (exact) solution
in a spacetime neighbourhood of $i^0$ and $\mathscr{I}^{\pm}$ as that of Lemma
\eqref{Lemma:Sol_Jacobi_and_Logs}.  Whether such solution can indeed
be obtained ---say, with a similar method to that of \cite{MinMacVal22}---
hence  justifying the heuristics leading to equation
\eqref{eq:logToPower}, will be discussed elsewhere.

\section{Conclusions}

In this article the NP constants for a spin-0 field propagating close
to spatial and null infinity are computed in terms of prescribed
analytic initial data close to $i^0$. This was performed for a spin-0
field propagating in Minkowski spacetime and hence an infinite
hierarchy of conserved charges were obtained.  To do so, the framework
of the $i^0$-cylinder was employed. The relation between the NP
constants at future and past null infinity was investigated by
identifying the part of the initial data that determines the NP
constants. As discussed in the main text ---see also
\cite{MinMacVal22, DuaFenGasHil22b}, even if analytic initial data
close to $i^0$ is prescribed, the solution looses regularity at the
critical sets $I^{\pm}$ where $i^0$ and $\mathscr{I}$ meet. This is
controlled by a constant $D_{p,p,m}$ appearing in the parametrisation
of the initial data. It was shown that, for initial data that does not
satisfy the regularity condition (namely initial data with $D_{p,p,m}
\neq 0$) the classical NP constants are not well-defined. If the
regularity condition is satisfied (i.e. $D_{p,p,m} = 0$), the
classical NP constant at $\mathscr{I}^{\pm}$ stem from independent
parts of the initial data ---parametrised by the constants
$A_{\ell+1,\ell,m}$ and $B_{\ell+1;\ell,m}$, respectively--- and hence
there is no correspondence between them. Nonetheless, it was also
shown that using the $f(\tilde{\rho})$-modified NP constants of
\cite{Keh21_a} for $f(\tilde{\rho})=\tilde{\rho}$ gives rise to
conserved quantities for which the regularity condition is not needed
and, in fact, correspond precisely the terms in the initial data
controlling the regularity of the field: $D_{\ell,\ell,m}$.  In
summary, one concludes that, unlike the classical NP constants, the
\emph{$i^0$-cylinder NP constants} at $\mathscr{I}^{\pm}$, up-to a
numerical factor, coincide as they stem from the same piece of initial
data. Finally, it was shown how the main identities used in the
definition of the NP constants provide a path for a modification of
the asymptotic-system heuristic expansion of~\cite{DuaFenGasHil21}
which is sensitive to the presence of the (higher-order)
$i^0$-cylinder-logs, at least, in flat spacetime.  Generalisations of
this heuristic approach for asymptotically flat spacetimes in a
similar approach to that of ~\cite{DuaFenGasHil21} will be left for
future work.

\subsection*{Acknowledgements}
 
EG holds a FCT (Portugal) investigator grant 2020.03845.CEECIND.  EG
acknowledges funding of the Exploratory Research Project
2022.01390.PTDC by FCT.  We have profited from scientific discussions
with D. Hilditch, J. Feng and M. Duarte.

\newpage

\appendix

\section{Conservation laws}\label{App:A}

The generalisation of the conservation laws \eqref{eq:cons_laws} in a
Schwarzschild background have been obtained in \cite{Keh21_a},
consequently the flat space version follows immediately from them.
Here, for convenience of the reader and to have a self-contained
discussion, the conservation laws in flat space are derived with the
addition of a source term.

\begin{proposition}\label{prop:main_commutation}
Let $\tilde{\phi}$ be a solution to
$\tilde{\square}\tilde{\phi}=\tilde{S}$.  Let $\phi =
\tilde{\rho}\tilde{\phi}$ and $S=\tilde{\rho}^3\tilde{S}$ then $\phi$
  satisfies
 \begin{align}\label{eq:main_commutation}
   \uL L (\bme^{+})^n \phi = -\frac{2n}{\tilde{\rho}}L (\bme^{+})^n
   \phi + \frac{1}{\tilde{\rho}^2}(n(n+1) + \Delta_{\mathbb{S}^2})(\bme^+)^n \phi -
   \frac{1}{\tilde{\rho}^2}(\bme^+)^nS.
 \end{align}
\end{proposition}
\begin{proof}
  This identity can be proven by induction. In terms of $\uL$ and $L$
  and $\phi$, the equation $\tilde{\square}\tilde{\phi}=\tilde{S}$
  simply reads
  \begin{align}\label{eq:uLLphiToSource}
    \uL L \phi = \frac{1}{\tilde{\rho}^2}(\Delta_{\mathbb{S}^2} \phi - S).
\end{align}
  Equation \eqref{eq:uLLphiToSource} corresponds case $n=0$ of
  expression \eqref{eq:main_commutation} and the basis of induction.
  Let's assume expression \eqref{eq:main_commutation} is valid for
  $n=N$.  For the induction step, one computes $\uL L (\bme^+)^{n+1}
  \phi$ as follows:
  \begin{flalign}
    \uL L (\bme^+)& ^{N+1} \phi = L \uL (\bme^+)^{N+1} \phi = L (\uL(
    (\tilde{\rho}^2L)(\bme^+)^{N} \phi)) \nonumber \\ & = L
    (\tilde{\rho}^2 \uL L (\bme^{+})^N \phi) + L(-2\tilde{\rho}L
    (\bme^{+})^N \phi) \nonumber \\ & = L (-2N \tilde{\rho} L
    (\bme^{+})^N \phi + (N(N+1) + \Delta_{\mathbb{S}^2}) (\bme^{+})^N \phi -
    (\bme^{+})^N S ) + L(-2\tilde{\rho}L (\bme^{+})^N \phi) \nonumber
    \\ & = -\frac{2(N+1)}{\tilde{\rho}}L (\bme^{+})^{N+1} \phi +
    \frac{1}{\tilde{\rho}^2}((N+1)(N+2) + \Delta_{\mathbb{S}^2})(\bme^+)^{N+1} \phi -
    \frac{1}{\tilde{\rho}^2}(\bme^+)^{N+1}S. \label{eq:inductionstep_commutation}
  \end{flalign}
   where the induction hypothesis (case $n=N$) was employed in the
   third line.  Noticing that equation
   \eqref{eq:inductionstep_commutation} is the same as expression
   \eqref{eq:main_commutation} with $n=N+1$ finishes the proof.
\end{proof}
Additionally, observe that since this calculation is done in flat
spacetime hence trivial commutations of $\uL$, $L$ and $\Delta_{\mathbb{S}^2}$ were
made in the first and fourth lines without the introduction of further
terms.
Using Proposition
\ref{prop:main_commutation} for $\ell=n$ and rearranging gives the
following:
\begin{corollary}\label{coro:main_commutation}
Let $\tilde{\phi}$ be a solution to
$\tilde{\square}\tilde{\phi}=\tilde{S}$ and let $\phi =
\tilde{\rho}\tilde{\phi}$ and $S=\tilde{\rho}^3\tilde{S}$.
then
\begin{align}\label{eq:coromain_commutation}
\uL (\tilde{\rho}^{-2\ell} L (\bme^{+})^{\ell}\phi_{\ell m}) =
 -\tilde{\rho}^{-2(\ell +1)}(\bme^+)^{\ell}S_{\ell m}
\end{align}
where $\phi_{\ell m}= \int_{\mathbb{S}^2} \phi \; Y_{\ell m} \;
d\sigma$ and $S_{\ell m}= \int_{\mathbb{S}^2} S \; Y_{\ell m}
\; d\sigma$ with $d\sigma$ denoting the area element in
$\mathbb{S}^2$.
\end{corollary}
Then, \emph{mutatis mutandis}, the time-reversed version of Proposition
\ref{prop:main_commutation} and Corollary \ref{coro:main_commutation} read
\begin{proposition}\label{prop:main_commutation_minus}
  Let $\tilde{\phi}$ be a solution to
  $\tilde{\square}\tilde{\phi}=\tilde{S}$.  Let $\phi =
  \tilde{\rho}\tilde{\phi}$ and $S=\tilde{\rho}^{3}\tilde{S}$ then $\phi$
  satisfies
 \begin{align}\label{eq:main_commutation_minus}
   L \uL (\bmue^{-})^n \phi = -\frac{2n}{\tilde{\rho}}\uL
   (\bmue^{-})^n \phi + \frac{1}{\tilde{\rho}^2}(n(n+1) +
   \Delta_{\mathbb{S}^2})(\bmue^-)^n \phi - \frac{1}{\tilde{\rho}^2}(\bmue^-)^nS.
 \end{align}
\end{proposition}
  \begin{corollary}\label{coro:main_commutation_minus}
Let $\tilde{\phi}$ be a solution to
$\tilde{\square}\tilde{\phi}=\tilde{S}$ and let $\phi =
\tilde{\rho}\tilde{\phi}$ and $S=\tilde{\rho}^3\tilde{S}$.  Then
\begin{align}\label{eq:coromain_commutation_minus}
L (\tilde{\rho}^{-2\ell} \uL (\bmue^{-})^{\ell}\phi_{\ell m}) =
-\tilde{\rho}^{-2(\ell +1)}(\bmue^-)^{\ell}S_{\ell m}
\end{align}
where $\phi_{\ell m}= \int_{\mathbb{S}^2} \phi \; Y_{\ell m} \;
d\sigma$ and $S_{\ell m}= \int_{\mathbb{S}^2} S \; Y_{\ell m} \;
d\sigma$ with $d\sigma$ denoting the area element in $\mathbb{S}^2$.
\end{corollary}
Notice that Propositions \ref{prop:main_commutation} and
\ref{prop:main_commutation_minus} along with their respective Corollaries
\ref{coro:main_commutation} and \ref{coro:main_commutation_minus}
could be employed for the analysis of the good-bad-ugly system in flat
spacetime of \cite{DuaFenGasHil22b}  by specifying the
source term $\tilde{S}$.  Such analysis and its generalisation to
asymptotically flat spacetimes will be discussed elsewhere.

\begin{IndRemark}\label{remark:sources}
  For the good-bad-ugly system in flat spacetime of
  \cite{DuaFenGasHil22b} one has
\begin{align}
       \tilde{S}_g =0 &\implies
       S_g=0,\\ \tilde{S}_{b}=(\partial_{\tilde{t}}\tilde{\phi})^2
       &\implies S_{b}= \frac{1}{4}\tilde{\rho}((L + \uL
       )\phi_g)^2,\\ \tilde{S}_{u}= 2\tilde{\rho}^{-2}
       \partial_{\tilde{t}}\tilde{\phi}_{u} &\implies S_u=
       (L+\uL)\phi_u,
\end{align}
and Corollaries \ref{coro:main_commutation} and
\ref{coro:main_commutation_minus} apply accordingly.
\end{IndRemark}


\normalem
\bibliographystyle{unsrt}

\end{document}